\documentclass{article}[12 pt, a4paper]

\usepackage{graphicx}
\pdfoutput=1 
\usepackage{mathptmx}      
\usepackage{amssymb}
\usepackage{amsmath}
\usepackage{mathrsfs}
\usepackage{latexsym}
\usepackage{amsthm}
\usepackage{bbm}
\usepackage{cite}
\title{A copula based Markov Reward approach to the credit spread in European Union. }
\author{Guglielmo D'Amico (1), Filippo Petroni (2), Philippe Regnault (3),\\ Stefania Scocchera(1)\footnote{corresponding author: Stefania Scocchera, e-mail: stefania.scocchera@unich.it},  Loriano Storchi(1) }
\date{ \footnotesize 1 Department of Pharmacy, University G. D'Annunzio Chieti - Pescara, Italy\\
2. Department of Management, University Politecnica delle Marche, Ancona, Italy\\
3. Laboratory of Mathematics, University of Reims Champagne-Ardenne, Reims, France.}

\newtheorem{Def}{Definition}
\newtheorem{remark}[Def]{Remark}

\newtheorem{Assumption}{Assumption}
\newtheorem{proposition}[Def]{Proposition}

\newcommand{\N}{\mathbb{N}}

\renewcommand{\P}{\mathbb{P}}
\newcommand{\E}{\mathbb{E}}

\newcommand{\bfP}{\mathbf{P}} 

\newcommand{\D}{DT(\textbf{sh}(t))} 

\begin{document}
\maketitle
\begin{abstract}
In this paper, we propose a  methodology based on piecewise homogeneous Markov chain for credit ratings and a multivariate model of the credit spreads to evaluate the financial risk in European Union (EU). Two main aspects are considered: how the financial risk is distributed among the European countries and how large is the  value of the total risk. The first aspect is evaluated by means of the expected value of a dynamic entropy measure. The second one is solved by computing the evolution of the total credit spread over time. 
Moreover, the covariance between countries' total spread allows understand any contagions in the EU. 
The methodology is applied to real data of $24$ European countries for the three major rating agencies: Moody's, Standard \& Poor's and Fitch. Obtained results suggest that both the financial risk inequality and the value of the total risk increase over time at a different rate depending on the rating agency and that the dependence structure is characterized by a strong correlation between most of European countries. 

\noindent keywords: Sovereign credit rating, Markov process, Dynamic measure of inequality, Copula,  Change-point
\end{abstract}
\section{Introduction}
\label{intro}
The interest on the sovereign securities has increased  after the occurrence of some sovereign defaults and financial crisis. In particular, the Eurozone has become of main interest, considering the economic and financial implications given by the integration of countries. 
Following our previous works (\cite{continuous}-\cite{discrete}) we aim at gaining insight on one of the financial implications  arising from the economic union of European countries. Specifically, we focus on the financial risk related to each country and to the European Union as a whole.
We refer to financial  risk as the countries' ability to face with their financial obligations. It is expressed by the amount of credit spread, that depends on the sovereign credit rating assignment. In particular, there are two main questions we want to reply. The first one concerns the distribution of this risk among European countries. The second one refers to the size of the risk faced by the whole set of countries. Therefore, the purpose of the present work is to understand the behaviour of the financial risk focusing both on the evolution of its total size and the assessment of the inequality of the risk distribution among countries. 
The European Union has been analysed focusing on different problems. As a matter of example the public debt and its ownership is investigated in \cite{eu}; in \cite{cds} the authors studied dependence of default risk of several Eurozone countries. While the problem about the income inequality has been faced in \cite{d2012}. Changes in European structure given by the exit of some members is analysed in \cite{grexit} and  \cite{brexit}.
The influence of rating dynamics on the credit spread evolution has been highlighted in financial literature by \cite{semi:spread} and \cite{huang}, mainly concerning industry sector. Another source of dependence for the credit spread evolution can be found in \cite{bivariate}. In this work the authors proposed a bivariate semi-Markov reward approach to include the counterpart credit risk. Others works proposed the application of a Copula to capture dependencies see \cite{copula4},  \cite{copula2}, \cite{copula1} and \cite{copula3}. However, these applications do not concern credit spread modelling. 
Regarding the credit rating studies,  rating modelling includes the Markov processes (see \cite{bangia2002}-\cite{belkin98}-\cite{JS}-\cite{nickell00}) and the semi-Markov processes (see, among others, \cite{d2017semi}, \cite{d2006homogeneous} amd \cite{semi:application}). Furthermore, sovereign credit ratings have been modelled by means of Markov processes (see, e.g.\cite{hu2002}, \cite{fuertes_sovereign}, \cite{sovereign2015a} ).
Information theory is also applied to economic and financial issues. In particular, dynamic inequality measures can be found in \cite{dr}  where the authors proposed a dynamic extension of common poverty indices. In \cite{d2012} the authors proposed a dynamic measurement of income inequality based on the Theil index, and successively a decomposition of this measure has been advanced in \cite{d2014}. The measure of inequality we propose in order to evaluate the financial risk is based on the last two contributions. 
Interestingly, the topic we are working on has never been faced in financial literature. We give a contribution on the existent literature by proposing a copula based Markov reward approach to model credit spread dynamics and to evaluate the financial risk. In particular, the questions we  posed are solved by the assessment of the dynamic Theil index for the financial inequality measurement and by the computation of the total credit spread for the quantification of the total risk in European Union. Furthermore, an analysis of the dependence structure is carried out by means of the covariance between countries' total credit spread.
The model has been implemented for three rating agencies: Moody's, Standard \& Poor's and Fitch, to find out if there are  differences stemming from the dissimilarities of the rating assignment process. \\
Obtained results show that: the financial inequality is increasing in the future for all three agencies, although the values evolve differently over time.  The total financial risk also increases with dissimilarities depending on the rating agency; the dependence structure is characterized by a strong correlation between countries. 
\noindent The paper is organized as follows: the second Section analyse the data while the third one  describes the model. Section \ref{theil:index}  presents the indicators computed to evaluate financial risk, i.e. the dynamic inequality measure, the total credit spread and the covariance between countries' total credit spread. In Section \ref{results} empirical results are discussed, followed by the concluding remarks. 

\section{Data analysis}
\label{data:analysis}
Our research objective is to provide a model able to measure inequality in the financial risk in a set of countries (or financial entities) and to assess the evolution in time of the total risk. To this end, we focused our attention on European countries and we collected data on two main financial variables: the sovereign credit ratings and the credit spreads.  
Sovereign credit rating is an ordinal measure of the country's credit risk. It expresses the ability of a country to face its financial commitments. Credit spreads are also collected, they are the difference between interest rates of various countries. It is well known that credit spread depends on sovereign credit rating assignment, see e.g. \cite{discrete}.

The sovereign credit rating assigned to the European countries by the three major rating agencies, Moody's, Standard \& Poor's and Fitch, has been collected.  Thus, we built three different datasets, one for each rating agency, collecting rating histories from November $23$, $1998$ to June $26$, $2018$ on a daily scale. The data are gathered from the \emph{Tradingeconomics} website and grouped into eight rating classes as shown in Table \ref{ratass}.
\begin{table}[!ht]
\centering
\begin{tabular}{|ccccccccc|}
\hline
Moody's&Aaa& Aa&A&Baa&Ba&B&Caa-Ca&C\\
\hline
S\&P&AAA&AA&A&BBB&BB&B&CCC-CC-C&SD-D\\
\hline
Fitch&AAA&AA&A&BBB&BB&B&CCC-CC-C&RD-D\\
\hline
rank&1&2&3&4&5&6&7&8\\
\hline
\end{tabular}
\caption{Rating class classification for each rating agency}\label{ratass}
\end{table}  
The rating class  $1$ is the best rating assignment meaning that the  issuer has an exceptionally strong capacity to cope with its financial commitments. Lower  credit rating assignments, i.e. $2, ... , 7$, imply the belief that the issuer is gradually less able to face with its financial commitments. Rank 8 denotes financial default.
As not all data were available for all European countries the sample is composed by $24$ members: Belgium, Bulgaria, Czech Republic, Germany, Denmark, Ireland, Greece, Spain, France, Croatia, Italy, Lithuania, Hungary, Malta, Netherlands, Austria, Poland, Portugal, Romania, Slovenia, Slovakia, Finland, Sweden, United Kingdom.\\
\begin{table}[!ht]\label{down:up} 
\centering
\begin{tabular}{|llccc|}
\hline
&&1998/2007&2008/2014&2015/2018\\
\hline
S\&P& upgrades&55.18&10.34&34.48\\
&     downgrades&6.25&78.12&15.63\\
\hline
Moody's&upgrades&66.67&4.16&19.17\\
&       downgrades&0&88.89&11.11\\
\hline
Fitch& upgrades&64.29&10.71&25\\
&     downgrades&12.5&78.13&9.37\\
\hline
\end{tabular}
\caption{upgrades/downgrade distribution over time for all rating agencies ( \% )}
\end{table}
Rating assignments are almost stables, in fact there are few transitions. We observed 61, 51, 60 transitions for S\&P, Moody's and Fitch, respectively. In particular, the upgrades / downgrades (transition to a better / lower rating class) are concentrated in different periods. Table \ref{down:up} illustrates the percentage of upgrades/downgrades experienced by all countries according to the three rating agencies assignment in three sub-periods: 1998-2007, 2008-2014, 2015-2018. According to Table \ref{down:up}, the downgrades are mostly concentrated in the period ranging between $2008$ and $2014$ covering the financial crisis  and the Greek crisis: 78,12\% for S\&P, 88.89\% for Moody's and 78.13\% for Fitch. On the other hand, the $55.18\%$, $66.67\%$ and $64.29\%$ of upgraded is observed before the 2008 for S\&P, Moody's and Fitch, respectively. While the rest of upgrades are mostly concentrated in the period spreading between the 2015 and 2018. Figure \ref{up:down} shows the timing of the upgrades and the downgrades for Moody's. As it is possible to see, before the 2008 there were no downgrades and the percentage of upgrades during the 2008-2014 is smaller than the other agencies. 
\begin{figure}[!h]
\centering
\includegraphics[width=0.8\textwidth]{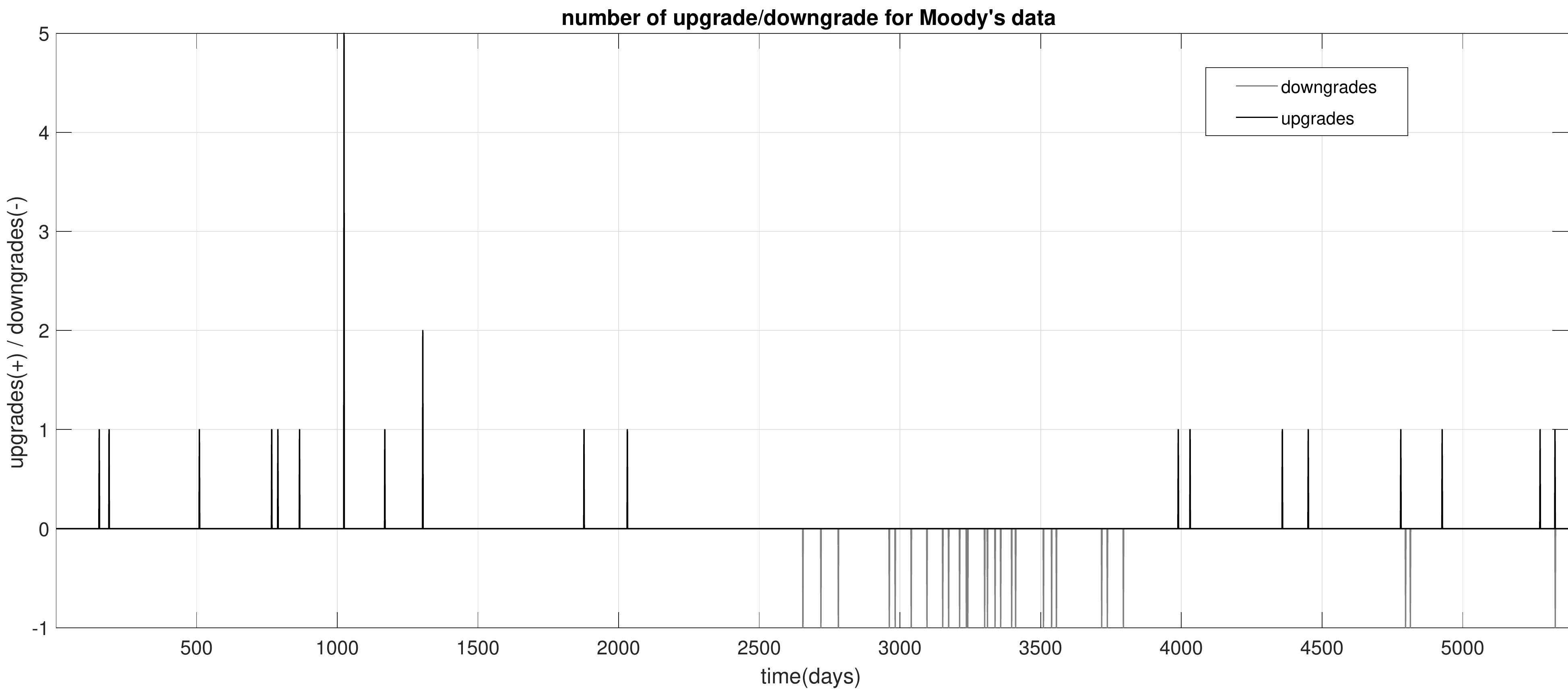}
\caption{number of upgrades/downgrades observed over Moody's data.}\label{up:down}
\end{figure}

The second variable we need is the credit spread.
Generally, the credit spread is given by the difference between the interest rate and a benchmark.
The benchmark  we use to compute the credit spread is the minimum value among the interest rates paid by all European countries.   
The reason underlying this choice is  that there are some countries experiencing lower interest rates than Germany. Consequently, it allows to avoid negative spreads. 
Thus, the credit spread can be interpreted as the premium for the higher risk paid by a given country compared to the ideal situation where the country pays the minimum value. 
Therefore, the long-term interest rate of sovereign government bonds are collected from the \emph{investing.com} web-site, on a daily scale. The data are available for all countries only starting from April, $26$, $2010$. Thus, the observation period ranges between this date and June $26$, 2018.
Figure \ref{rir} shows the evolution of the interest rate and the credit spread paid in European Union over the observed period. After a peak around the end of $2011$ and the start of $2012$, both indicators show a decreasing trend. The difference between them is relevant until the 2012 and then it decreases, suggesting that the minimum value has decreased over time. This difference is ,in fact 
 equal to the minimum value of interest rates multiplied by the number of considered countries.  

\begin{figure}[!ht]
\centering
\includegraphics[width=0.8\textwidth]{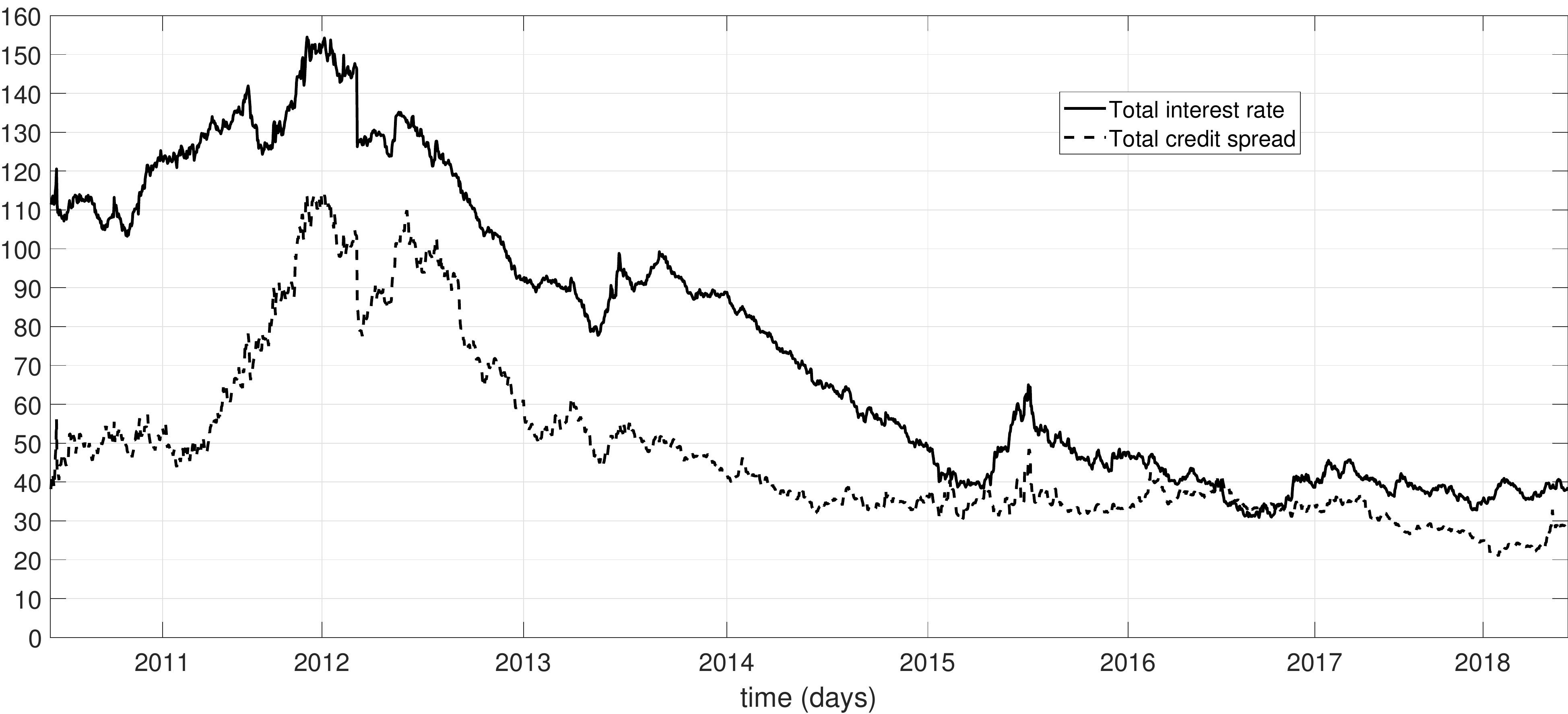}
\caption{observed evolution of the total credit spread and total interest rate}\label{rir}
\end{figure}
\begin{figure}[!ht]
\begin{minipage}{1\textwidth}
\centering
\includegraphics[width=0.55\textwidth]{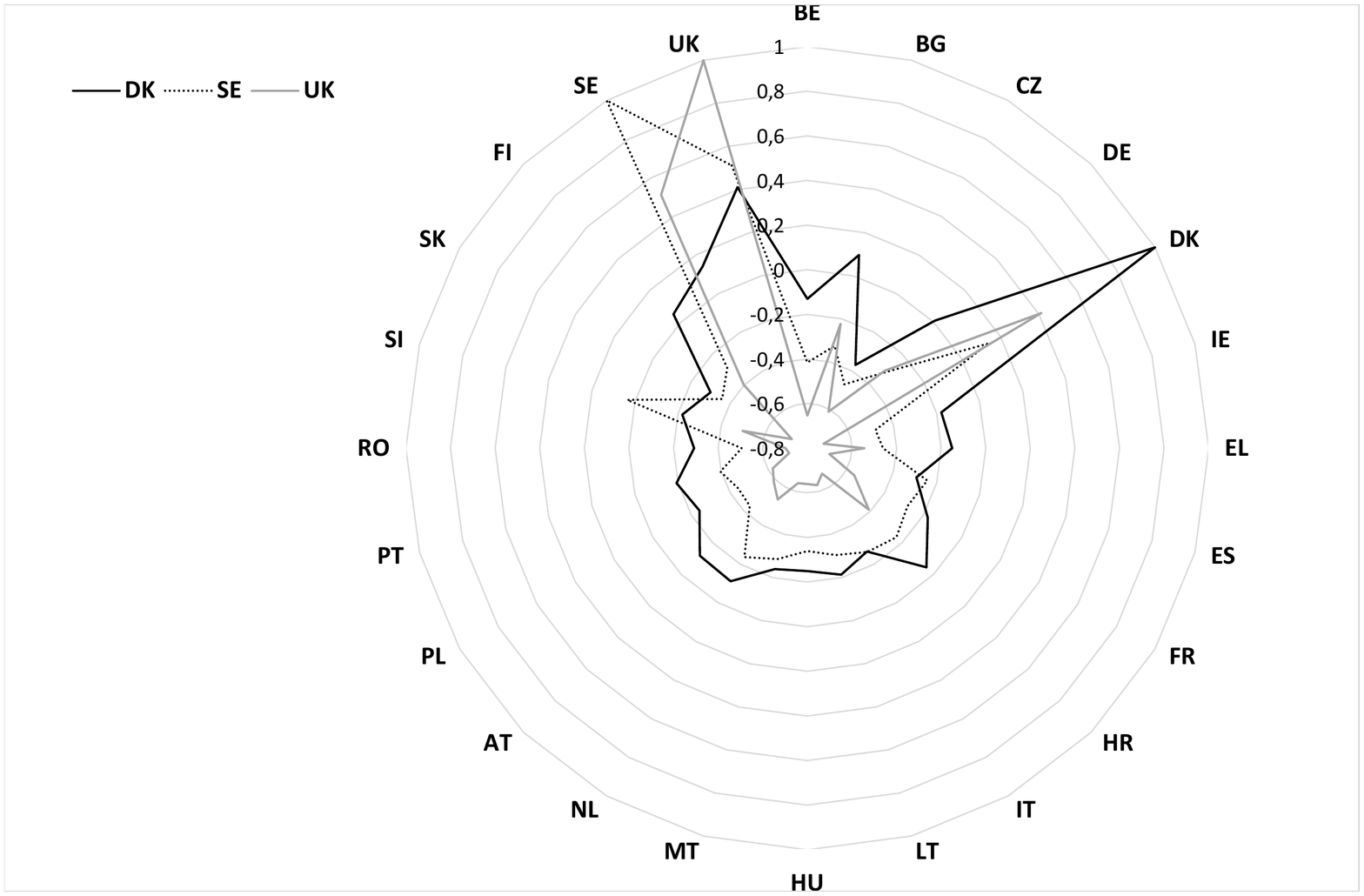}
\caption{Correlation coefficient of the credit spread's time series between Denmark (black line), Sweden (dashed line), United Kingdom (gray line) and all other countries.  BE: Belgium, BG: Bulgaria, CZ: Czech Republic, DE: Germany, DK: Denmark, IE: Ireland, EL: Greece, ES: Spain, FR: France, HR: Croatia, IT: Italy, LT: Lithuania, HU: Hungary, MT: Malta, NL: Netherlands, AT: Austria PL: Poland, PT: Portugal, RO: Romania, SI: Slovenia, SK: Slovakia, FI: Finland, SE: Sweden, UK: United Kingdom}\label{negative:correlation}
\end{minipage}
\vspace{0.2cm}
\begin{minipage}{0.49\textwidth}
\includegraphics[width=1.1\textwidth]{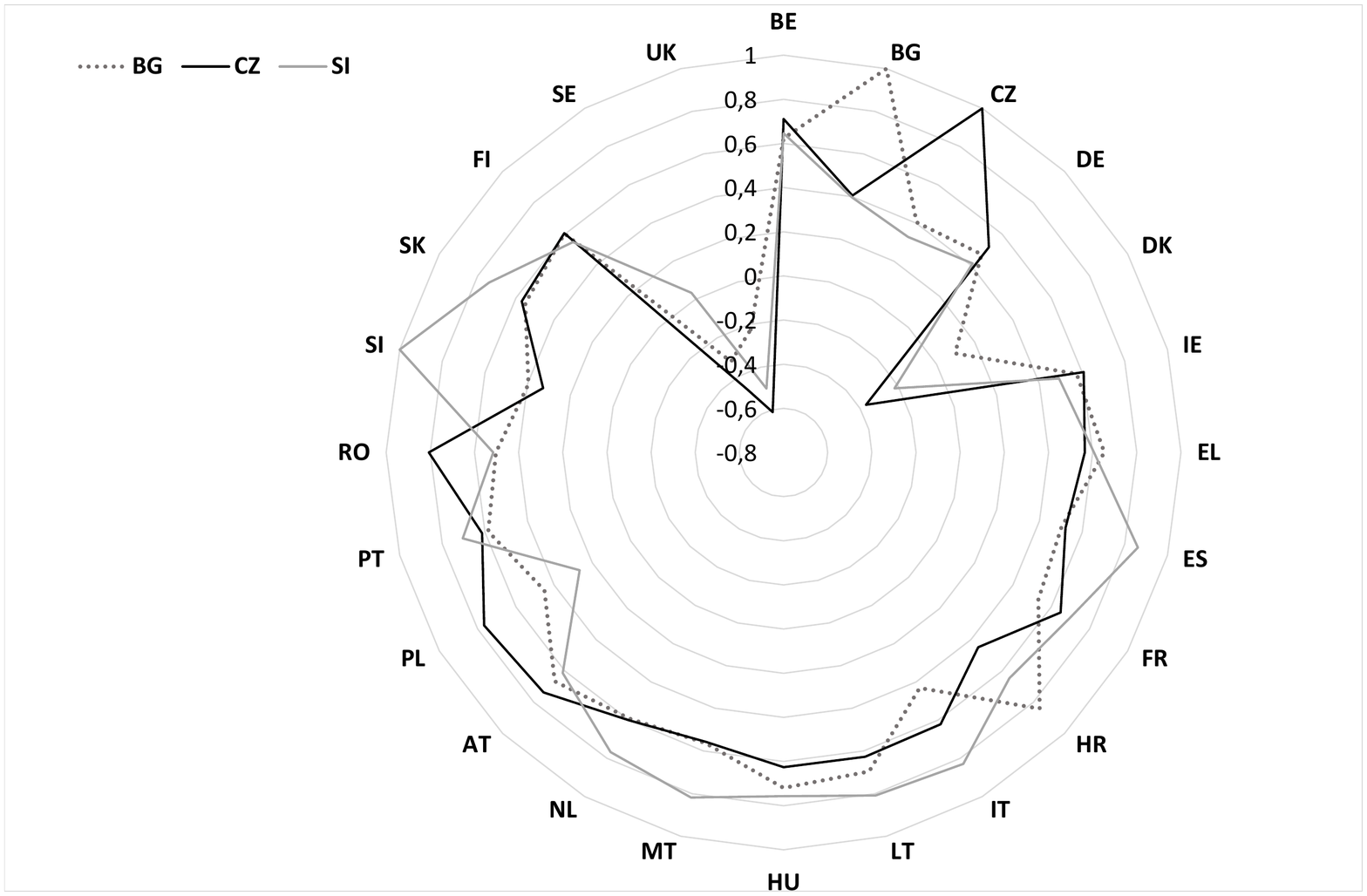}
\end{minipage}
\hspace{0.05cm}
\begin{minipage}{0.49\textwidth}
\includegraphics[width=1.1\textwidth]{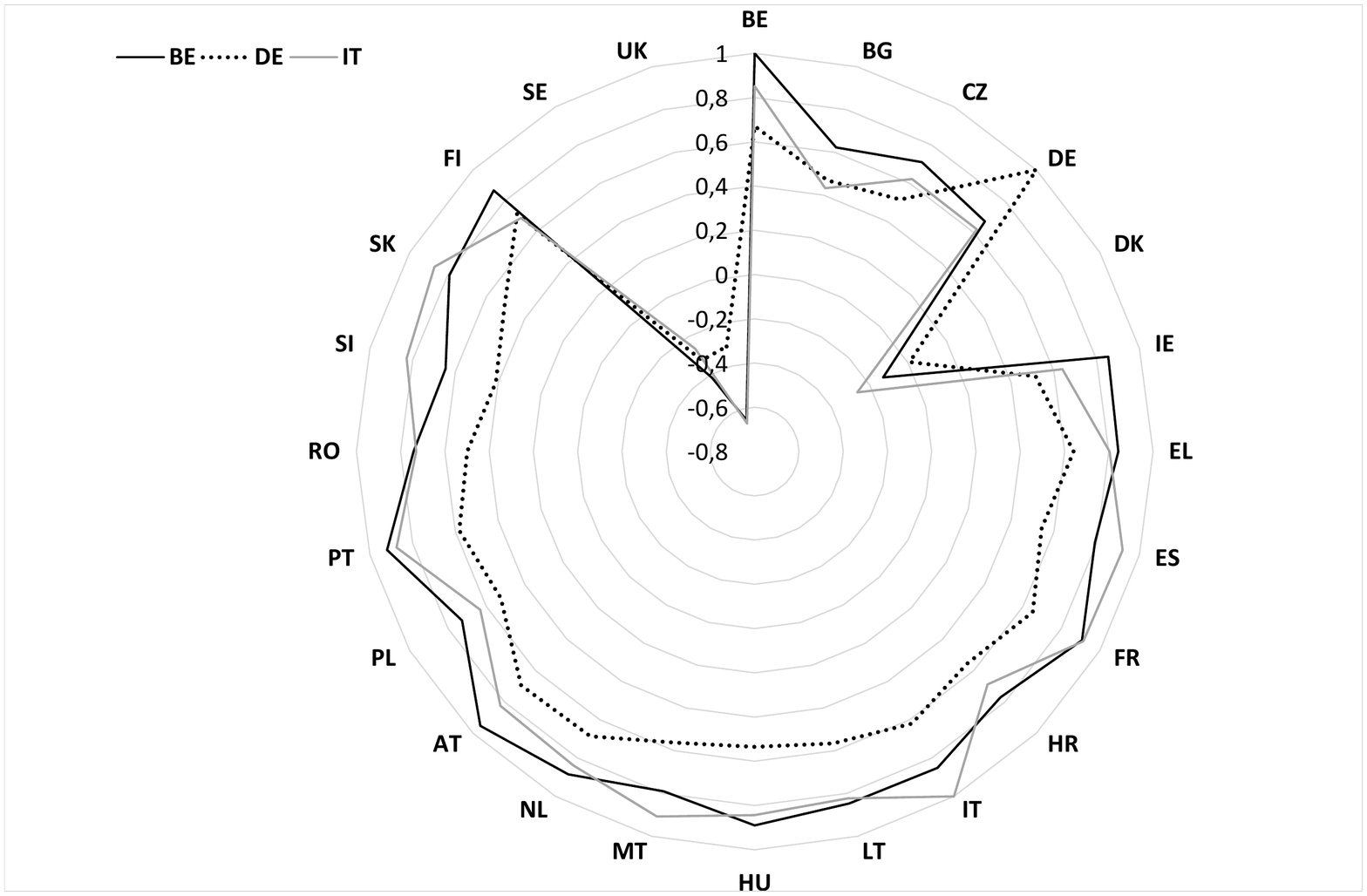}
\end{minipage}
\caption{Left panel: Correlation coefficient of the credit spread's time series between Bulgaria (dashed line), Czech Republic (black line), Slovenia (gray line) and all other countries. Right panel: correlation coefficient between Belgium (black line), Germany (dotted line), Italy (gray line) and all other countries.}\label{positive:correlation} 
\end{figure}
Moreover, analysis on the spread times series shows that European countries are positively correlated, being the coefficient of correlation between most of countries higher than $0,5$. However, exceptions are represented by Denmark, Sweden and United Kingdom. According to Figure \ref{negative:correlation}, these countries are negatively correlated with respect to all other countries, but they are positively correlated between each other.
Figure \ref{positive:correlation} shows the correlation coefficient estimated between  2 different sub-sets of countries and the whole set. In the left panel Bulgaria, Czech Republic and Slovenia are represented, while in the right one Belgium, Germany and Italy are shown\footnote{It is worth noting that we represent only  this sub-sets as they are representative of the whole group. The rest of the countries have a correlation structure which is similar to that of Belgium, Italy and Germany and that of Bulgaria, Czech Republic and Slovenia. Whereas the first group has a different structure.}. 
Bulgaria, Czech Republic and Slovenia are positively correlated with all other countries (except for Denmark, Sweden and United Kingdom) as the second group. However, the positive / negative  correlation of the first group is smaller / higher than the second one. 

The peculiarities of the data shown above push us to propose a model that generalizes that of our previous works (see \cite{discrete} and \cite{continuous}). The ratings dynamics are modelled according to a piecewise homogeneous Markov chain in order to include changes in rating dynamics and to explain the different behaviour illustrated in Table \ref{up:down}. 
Moreover, the inclusion of a stochastic process describing the spread evolution and its dependence among countries allows for the high correlation among countries. Details about the proposed methodology are given in the next section.

\section{Modelling rating and credit spread dynamics}
\label{model}
Let $\mathcal{C}$ be a set of $N$ countries. At any time $t \in \N$ (discrete times can refer to months, weeks or days, depending on the granularity of the study), each country $c$ $ \in \mathcal{C}$  is rated by financial agencies;  its rating assignment $x^{c}(t)$ belongs to an ordered set ${E} :=\{1 , 2 , ...,  D \}$, where $D$ denotes the number of possible rating assignments. The rating assignment of a country aims at reflecting the financial viability of the country and is derived from the financial, economic, fiscal and political situations of the country\footnote{See the rating policies of the three agencies mentioned above on their web sites.}.
As usually done in credit rating studies, the sequences $x^{c}(t)$, $c \in \mathcal{C}$ are assumed to be realizations of stochastic processes $\textbf{X}^c := (X^c(t))_{t \in \N}$, $c \in \mathcal{C}$.

At the same time, a country $c$ can borrow money from financial institutions (or from the private investors) up to the application of an interest rate $\textrm{ir}^c(t)$. Again, the sequences $\textrm{ir}^c(t)$, $c \in \mathcal{C}$, are assumed to be realizations of stochastic processes $\textbf{IR}^c := (IR^c(t))_{t \in \N}$. From the interest rates, we derive the credit spreads $s^c(t)$, $c \in \mathcal{C}$ at any time $t$: for a country $c$, it is the difference between the interest rate ($ir^c(t)$) paid at time $t$ and the minimum value among all interest rates at the same time. Precisely,
\begin{equation}\label{cs}
s^c(t) := ir^c(t) - \underset{d \in \mathcal{C}}{\min}\{ir^d(t)\},
\end{equation}
or, in terms of processes, $\textbf{S}^c := \textbf{IR}^c - \underset{d \in \mathcal{C}}{\min}\textbf{IR}^d $.

The processes $\textbf{X}^c$ and $\textbf{S}^c$ evolve jointly: unfavourable situations for country $c$ yield low rating by agencies and high interest rate and possibly high credit spread (if the situation of the country paying the minimum value is less unfavourable). 

Let us now describe the modelling for the rating and credit spread processes. The following assumptions hold all along the paper.

\begin{Assumption} \label{Assiid}
 The processes $\textbf{X}^c$, $c \in \mathcal{C}$, are independent and identically distributed (i.i.d.). 
 
 In the following, we shall denote by $\textbf{X}$ a process drawn from their common distribution.
\end{Assumption}

Assumption~\ref{Assiid}  is motivated by theoretical and practical requirements: we need the model to be simple and flexible so as to perform estimation.  In addition, due to the sparsity of data, the estimation of a correlation structure between countries may be a difficult task. 

\begin{Assumption} \label{AssPiecewise}
 The process $\textbf{X}$ is a piecewise homogeneous Markov chain taking values in the ordered finite set $E = \{1, \dots , D\}$, i.e., there exist a positive number
$k \in \N^* \cup\{\infty\}$, a sequence $\tau_0 = 0 < \dots < \tau_k = \infty$ of increasing times and a sequence ${}^{(0)} \bfP, \dots, {}^{(k)}\bfP$ of stochastic matrices  such that for any $l \in \N, l \leq k$, for any $t \in \{\tau_l,\dots, \tau_{l+1}-1\}$ and any $x,y  \in E, \quad \textbf{x}_{0:{(t-1)}} = (x_0, \dots, x_{t-1}) \in E^t$, the following Markov property holds:
 $$\P \left( X(t+1) = y | X(t) = x,  X(0:(t-1)) = x_{0:(t-1)} \right) = \P \left( X(t+1) = y | X(t) =x \right) = {}^{(l)}p_{xy},$$
 where $X(0:(t-1)) = (X(0), \dots, X(t-1))$ and ${}^{(l)}p_{xy}$ denotes the transition probability  from $x$ to $y$ according to the matrix ${}^{(l)} \bfP$.
\end{Assumption}

Assumption~\ref{AssPiecewise} means that the rating process evolves according to a Markovian dynamic which may change in time. 
The Markovian assumption is widely used in financial literature concerning credit rating dynamics (see e.g. \cite{bangia2002}, \cite{belkin98}, \cite{JS}, \cite{nickell00}), most of all being concerned with sovereign credit ratings (see e.g.,\cite{fuertes_sovereign}, \cite{hu2002}, \cite{sovereign2015a} and \cite{wei_multi}). In our model, the homogeneity property is limited to some sub-periods composing the time line. This is motivated by the existence of some events that result in abrupt changes in financial situations of countries or rating policies (such as a financial crisis). The parameter $k$ in Assumption~\ref{AssPiecewise} represents the number of such abrupt changes. The time intervals $\{\tau_l, \dots, \tau_{l+1}-1\}$, for $l \in \{0, \dots, k-1\}$ correspond to the periods on which the rating dynamics are fixed and described by the transition matrix ${}^{(l)} \bfP$.
 The detection of  these abrupt changes ( also called change-points) for homogeneous Markov chain are studied in \cite{detect},  \cite{online:cp} and \cite{cp2012}, but it has been never applied to this financial problem.
\begin{Assumption} \label{AssCondDist}
 For any time $t \in \N$ and any country $c \in \mathcal{C}$, the conditional distribution of the credit spread $S^c(t)$ knowing $X^c(t) =x$, with $x \in E$, does not depend on $t$, nor $c$; we denote it $F_x$ and assume that $F_x$ is continuous, with density function $f_x$. Mathematically,
 \begin{equation} \label{EqnCondDistSpread}
 F_x := \mathcal{D}(S^c(t) | X^c(t) =x), \quad \textrm{for any } t \in \N, c\in \mathcal{C}.
 \end{equation}
 In the following, we will denote $W_x$ a random variable with distribution $F_x$; it will be substituted to $S^c(t)$, restricted to the event $X^c(t) = x$, when suitable, for  simplifying expressions of some conditional expectations such as $\mathbb{E}(S^c(t) | X^c(t) =x) = \E(W_x)$.
\end{Assumption}

Assumption~\ref{AssCondDist} stems from the recognition regarding the influence on the credit spread evolution of the rating dynamics. This is formalized through the assumption of a common spread distribution for countries with the same rating assignment.

\begin{Assumption} \label{AssCopula}
For any time $t \in \N$, the conditional joint distribution of $(S^1(t), \dots, S^N(t))$ knowing $(X^1(t) = x^1, \dots, X^N(t) = x^N)$, with $(x^1, \dots, x^N) \in E^N$ is given by
 $$\mathcal{D}(S^1(t), \dots, S^N(t) | X^1(t) = x^1, \dots, X^N(t) = x^N) =C_\theta (F_{x^1}, \dots, F_{x^N}),$$
 where $F_{x}$, $x \in E$ are given by~(\ref{EqnCondDistSpread}) and $C_\theta$ is a parametric copula, with dependence parameter $\theta$.
 
 In the following, for a given $N$-tuple $\textbf{x} = (x^1, \dots, x^N) \in E^N$, we will denote $\textbf{W}_{\textbf{x}}:=(W_{x^1}, \dots, W_{x^N})$ a random vector with distribution $C_\theta(F_{x^1}, \dots, F_{x^N})$.
\end{Assumption}

According to Assumption~\ref{AssCopula}, the multivariate stochastic process $(\textbf{S}^1, \dots, \textbf{S}^N)$ describing the credit spread evolution, is controlled by the Markov chains $\textbf{X}^c$, $c \in \mathcal{C}$, melt by some parametric copula. The use of the copula is justified by the existence of some dependences between the credit spreads of countries that the model has to render. 


\section{Financial risk  indicators}
\label{theil:index}
In this section we introduce the definition of the inequality measure advanced by Theil and its generalization for stochastic processes. Furthermore, two proxies of financial risk are presented: the  expected total credit spread paid by the whole group  and  the co-variance between the total credit spread paid by two countries. 
\subsection{Dynamic Theil index}
\label{dynamic:inequality}
One of the most important index used as a measure of inequality is the Theil index \cite{theil}. It is closely related to the Shannon entropy of a probability distribution \cite{shannon}. 

Given a probability distribution $p = (p^1, \dots, p^N)$ on a finite set with cardinal $N$ -- say $\{1, \dots, N\}$, Theil index $T(p)$ of $p$ is defined as the Kullback-Leibler (KL) divergence $\mathbb{K}(p|u)$ between ${p}$ and the uniform distribution $u$, or equivalently, as the difference between $\log(N)$ and the Shannon entropy $S({p})$. Precisely,
\begin{equation}\label{EqnStaticTheilIndex}
T(p) := \mathbb{K}(p|u) := \sum_{i=1}^N p^i \log(N \cdot p^i) = \log(N) - S(p), 
\end{equation}
where $S(p) = -\sum_{i = 1}^N p^i \log p^i$.

The definition of Theil index has been extended for stochastic processes by \cite{d2012}; see also \cite{d2014} for an additive decomposition of this index. Based on these references, we now introduce the dynamic Theil index of credit spreads, which we use to assess the inequality of the financial risk distribution. 

Let  the share of credit spread of a country $c \in \mathcal{C}$ at time $t \in \N$ be defined as the proportion of its credit spread $S^c(t)$ relative to the sum -- or total, of credit spreads of all countries $TS(t) :=  \sum_{d \in\mathcal{C}} S^d(t)$; mathematically,
\begin{displaymath}
 sh^c(t) = \frac{S^c(t)}{TS(t)} = \frac{S^c(t)}{ \displaystyle \sum_{d \in \mathcal{C}} S^d(t)}. 
\end{displaymath}
The vector of shares of credit spreads  at time $sh(t) := (sh^c(t))_{c \in \mathcal{C}}$ defines a probability distribution on the set of countries $\mathcal{C}$.  Note that $\textbf{sh} := (sh(t)_{t \in \N})$ is a stochastic process (taking its values in the set of probability distributions on $\mathcal{C}$) that depends on  the stochastic processes $\textbf{S}^c$, $c \in \mathcal{C}$.
We call dynamic Theil index of credit spread, the stochastic process $\D$, namely
\begin{equation} \label{EqnDynamicTheilIndex}
\D = \sum_{c \in \mathcal{C}} sh^c(t) \log(N \cdot sh^c(t)), \quad t \in \N.
\end{equation}

Both deterministic and dynamic Theil indices satisfy the following properties, that immediately stem from Shannon entropy's or KL divergence's (see e.g., \cite[Chapter 2]{cover}); they are formulated here for dynamic Theil entropy: 
\begin{enumerate}
\item $\D$ belongs to $[0,\log(N)]$ for any $t \in \N$, the lower and the upper bounds are achieved respectively  for the uniform distribution and for any Dirac measure. 
The lower bound corresponds to the perfect equi-distribution of financial risk as all countries pay the same amount of the credit spread. On the other hand, the upper bound refers to the concentration of the financial risk as one country pays the total spread. 
\item It is additively decomposable. Particularly, if we denote $g_x(t) := \{ c \in \mathcal{C} : X^c(t) = x  \}$ the subset of countries with rating assignment $x \in \{1, \dots, D\}$ at time $t \in \N$, the dynamic Theil index $\D$ is the sum of the Theil index of the distribution of share of credit spread by subset $q(t) = (q_1, \dots, q_D)$, where
$$q_x = \sum_{c \in g_x} sh^c, \quad x \in \{1, \dots, D\}$$
and of the average of Theil indices of subsets. Namely, dropping out the dependence in time in notations for sake of simplicity,
\begin{equation}\label{decomposed}
\D = DT(q) + \sum_{x =1}^D q_x \sum_{c \in g_x}  sh_{c|g_x} \log sh_{c|g_x}, 
\end{equation}
where $sh_{c|g_x} = \frac{sh_c}{q_x}$, for $c \in g_x$ represents the conditional distribution of share of credit spread of country $c$ knowing that $c$ belongs to subset $g_x$ (i.e., its rating assignment is $x$).
$DT(q)$ is the inequality measure between  rating classes $g_x$, $x = 1, \dots, D$ -- say the inter-class inequality measure, while $$ \sum_{x =1}^D q_x \sum_{c \in g_x}  sh_{c|g_x} \log sh_{c|g_x} = \sum_{x = 1}^D q_xDT(sh_{.|g_x}),$$ is a measure of the inequality within groups -- say the intra-group inequality measure. 

\item It is sensitive to the transfer of credit spreads in the population. This means that the Theil index is more sensitive to migration in the lower tail of the distribution of the credit spreads than it is to migration in the upper tail. 
\item It is invariant with respect to class permutation. This means that its value would be close to zero if all $ c  \in \mathcal{C}$ paid similar credit spreads, disregarding the amount of credit spreads, e.g. the class occupied by all countries. Furthermore, with increasing value of the credit spreads, $\D$ become smaller.
\end{enumerate}

The dynamic Theil index  is better summarized by computing the first order moment, $\E_0[\D]$.
\begin{proposition}
According to assumptions \textbf{A1 - A4}:
\begin{equation}\label{edt}\small 
\begin{aligned}
\E_0[\D] = &\sum_{i=1}^N \sum_{(a_1, \dots, a_N) \in E^N} \left[ \prod_{h=1}^N \sum_{(b_1, \dots, b_l)\in E^l} \prod_{d=0}^{l-1}  \, ^{(d)}P_{b_d, b_{d+1}}^{(\tau_{d+1}-\tau_d)} \cdot ^{(l)}P_{b_{d+1},a_h }^{(t - \tau_{d+1})} \right]\\
& \cdot \int_{0}^{+\infty} \cdots \int_{0}^{+\infty} \E \left(\frac{z_{a_i}}{\sum_{j=}^N z_{a_j}} \log N \frac{z_{a_i}}{\sum_{j=}^N z_{a_j}} \right) f_{(a_1, \dots, a_N)}\left(z_{a_1}, z_{a_2}, \dots, z_{a_N} | dz_{a_1}, dz_{a_2}, \dots, dz_{a_N}\right).
\end{aligned}
\end{equation}
where, 
\begin{equation}\small
f_{X^1(t) \dots X^{N}(t)} = \frac{\partial^N F_{X^1 \dots  X^{N}}}{\partial y_1 \cdots  \partial y_N} = c_{\theta}\left(F_{X^1(t)}(y_1),\dots ,  F_{X^N(t)}(y_N)\right) f_{X^1(t)}(y_1) \cdot \ldots \cdot f_{X^N(t)}(y_N).
\end{equation}
\end{proposition}
\begin{proof}
Let $(i_1, \dots , i_N)$ be the vector collecting the rating assignments of all $N$ countries at the initial time $t=0$ and let $(a_1, \dots, a_N)$ the vector collecting the rating assignment in a given future period t. The expected value is given by 
\begin{equation}\small\label{edt1}
\E_0\left[T(\textbf{p}(t))\right] = \E_0\left[\sum_{i=1}^N p^i (t)\log(N p^i(t))\right] = \sum_{i=1}^N \E \left[ \frac{S^i(t)}{\sum_{j=1}^N S^j(t)} \log \left(N \cdot \frac{S^i(t)}{\sum_{j=1}^N S^j(t)}\right) \right]
\end{equation}
\begin{equation}\small\label{edt2}
 = \sum_{i=1}^N \E \left[ \E \left[ \frac{S^i(t)}{\sum_{j=1}^N S^j(t)} \log \left(N \cdot \frac{S^i(t)}{\sum_{j=1}^N S^j(t)}\right) \bigg{|} \left(X^1(t), \dots , X^N(t) \right) \right]\right],
\end{equation}
According to the assumption \textbf{A3} (\ref{edt2}) becomes;
\begin{equation}\small\label{edt3}
\begin{aligned}
\E_0\left[\D \right] = &\sum_{i=1}^N \sum_{(a_1, \dots , a_N)\in E^N}\P\left(X_1(t)=a_1, \dots , X^N(t)=a_N | X^1(0)=i_1, \dots, X^N(0)=i_N \right)\\
&\cdot \E \left[ \frac{W_{a_i}}{\sum_{j=1}^N W_{a_j}} \log \left(N \cdot \frac{W_{a_i}}{\sum_{j=1}^N W_{a_j}}\right)\right]
\end{aligned}
\end{equation}
From the piecewise Markov chain assumption (\textbf{A2}) it follows that, given $t \in \N$,
\begin{equation}\small\label{edt4}
\exists \hspace{0.2cm} l \in \{ 0, 1, 2, \dots, k \} : t \in \{ \tau_l, \dots, \tau_{l+1} - 1 \}.
\end{equation}
and thus, that 
\begin{equation}\small\label{edt5}
\begin{aligned}
 & \P\left(X^h(t) = a_h | X^h(0) = i_h \right)= \\
 &= \sum_{b_1 \in E} \sum_{b_2 \in E} \cdots \sum_{b_l \in E} \P \left( X^h(t) = a_h, X^h(\tau_l) = b_l, \dots, X^h(\tau_2) = b_2, X^h(\tau_1) = b_1 | X^h(0) = i_0 \right),
\end{aligned}
\end{equation}
which can be rewritten as
 \begin{equation}\small\label{edt6}
 \begin{aligned}
   \P\left(X^h(t) = a_h | X^h(0) = i_h \right) 
 & = \sum_{(b_1, \dots, b_l) \in E} \P \left( X^h(t) = a_h | X^h(\tau_l) = b_l  \right) \cdot  \P \left( X^h(\tau_l) = b_l | X^h(\tau_{l-1}) = b_{l-1}\right) \cdot \dots\\
 & \cdot  \P \left( X^h(\tau_2) = b_2 | X^h(\tau_1) = b_1  \right) \cdot  \P \left( X^h(\tau_1) = b_1 | X^h(0) = i_0  \right) \\
  &= \sum_{(b_1, \dots, b_l) \in E} \hspace{0.1cm} ^{(l)} P^{(t-\tau_l)}_{b_l, a_h} \cdot ^{(l-1)} P^{(\tau_l-\tau_{l-1})}_{b_{l-1}, b_l } \cdot ... \cdot ^{(1)} P^{(\tau_2 - \tau_1)}_{b_1, b_2} \cdot ^{(0)} P^{(\tau_1)}_{i_0, b_1} \\
&= \sum_{(b_1, \dots, b_l) \in E} \prod_{d=0}^{l-1} \hspace{0.1cm}^{(d)} P^{(\tau_{d+1} - \tau_d)}_{b_d, b_{d+1}} \cdot ^{(l)} P^{(t - \tau_{d+1})}_{b_{d+1}, a_h}.
\end{aligned}
 \end{equation}
Therefore, by substitution of (\ref{edt6}) into (\ref{edt3}) we obtain
 \begin{equation}\label{edt8}\small 
\begin{aligned}
\E_0[\D] =& \sum_{i=1}^N \sum_{(a_1, \dots, a_N) \in E^N} \left[ \prod_{h=1}^N \sum_{(b_1, \dots, b_l)\in E^l} \prod_{d=0}^{l-1}  \, .^{(d)}P_{b_d, b_{d+1}}^{(\tau_{d+1}-\tau_{d})} \cdot \, ^{(l)}P_{b_{d+1},a_h }^{(t - \tau_{d+1})} \right]\\
& \cdot \E \left[ \frac{W_{a_i}}{\sum_{j=1}^N W_{a_j}} \log \left(N \cdot \frac{W_{a_i}}{\sum_{j=1}^N W_{a_j}}\right)\right].
\end{aligned}
\end{equation}
Finally, the last term of (\ref{edt8}) can be computed, under assumptions \textbf{A4}, according to:
\begin{equation}\small\label{edt9}
\begin{aligned}
  &\E \left[ \frac{W_{a_i}}{\sum_{j=1}^N W_{a_j}} \log \left(N \cdot \frac{W_{a_i}}{\sum_{j=1}^N W_{a_j}}\right)\right]  \\
&=\int_{0}^{+\infty} \cdots \int_{0}^{+\infty}  \left(\frac{z_{a_i}}{\sum_{j=}^N z_{a_j}} \log N \frac{z_{a_i}}{\sum_{j=}^N z_{a_j}} \right) f_{(a_1, \dots, a_N)}\left(z_{a_1}, z_{a_2}, \dots, z_{a_N} | dz_{a_1}, dz_{a_2}, \dots, dz_{a_N}\right).
\end{aligned}
\end{equation}
 and thus, (\ref{edt}) holds.
\end{proof}

\subsection{The total credit spread}
\label{total:spread}
As discussed in Section \ref{dynamic:inequality} the Theil index, both in its deterministic and dynamic formalization, is invariant with respect to class permutation. Thus, a measure of the total risk is needed to better understand and interpret the financial risk.
For seek of clarity, let make a very simple example. Suppose we have two situations where five countries pay, for four periods, the credit spread shown in Table \ref{example}.
\begin{table}[!ht]
\centering
\begin{tabular}{|c|ccccc|c|ccccc|}
\hline
case 1 &a&b&c&d&e&case 2&a&b&c&d&e\\
\hline
$t=1$&     2&     4&     5&     6&     3&  $t=1$ &  2&     4&     5&     6&     3\\
\hline
$t=2$&    12&    14&    15&    16&    13&  $t=2$ &  3&     5&     6&     7&     4\\
\hline
$t=3$&    22&    24&    25&    26&    23&  $t=3$ &  4&     6&     7&     8&     5\\
\hline
$t=4$&    32&    34&    35&    36&    33&  $t=4$ &  5&     7&     8&     9&     6 \\
\hline
\end{tabular}\caption{Credit spread paid by five agents for four periods.}\label{example}
\end{table}
The total spread paid by the sample in the two different situations is $TC_1=\{20, 70, 120, 170 \}$ and $TC_2= \{ 20, 25, 30, 35 \}$. 
\begin{figure}[!h]
\centering
\includegraphics[width=0.7\textwidth]{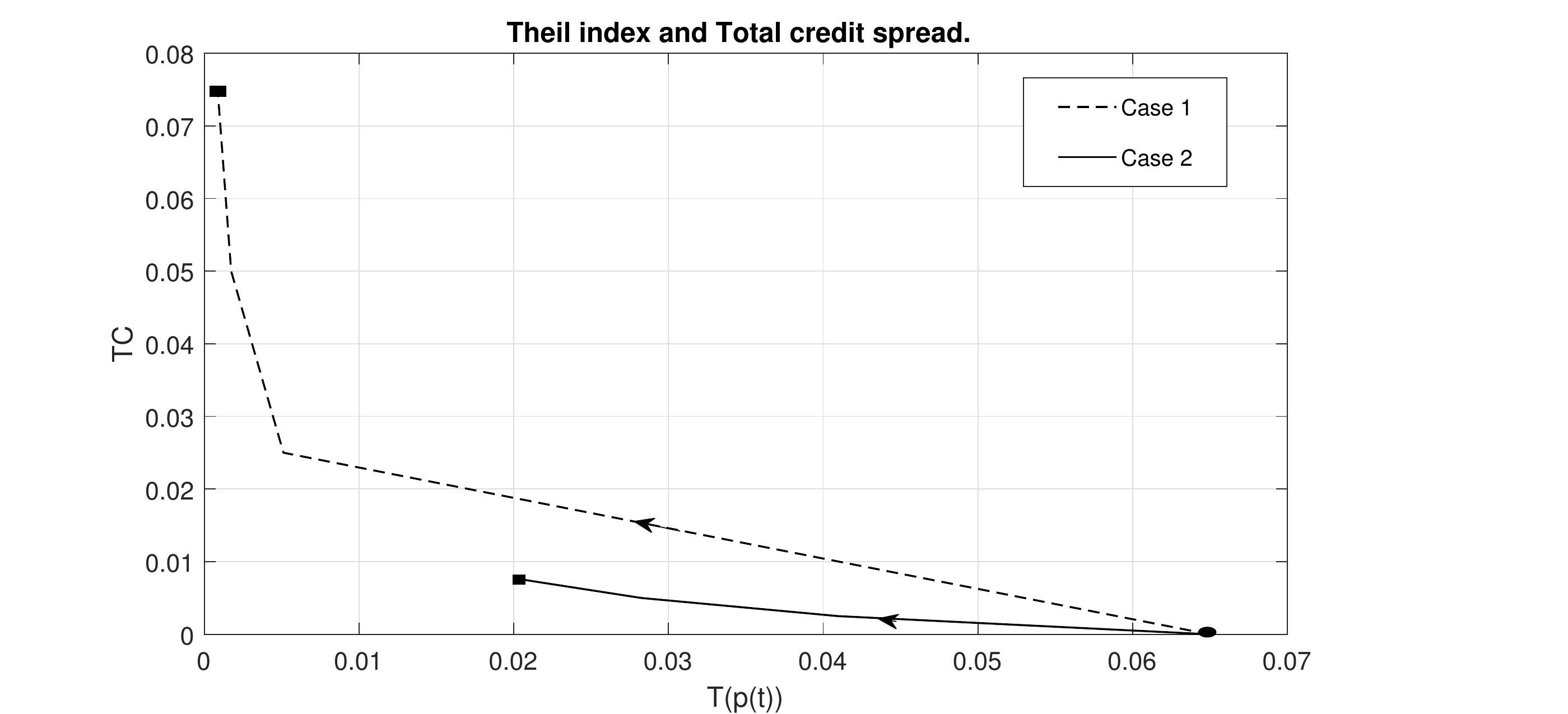}
\caption{Theil index and Total spread computed for the first case -dashed line- and for the second one -continuous line}\label{tdelta}
\end{figure}
The resulting inequality measure, computed according to Equation (\ref{EqnStaticTheilIndex}), is shown in the Figure \ref{tdelta}, along with the total spread relative variation.
Although the differences among countries are the same in both cases, the value of the Theil index is lower in the first case, corresponding to the higher values of the total spread. By considering only the inequality measure, the results  suggest that the best situation is that regarding the first case. For instance, the indices start from the same value, i.e. 0.065 corresponding to the black circle in the figure, but they evolve differently, reaching 0.0009 for the first case and 0.02 in the second case (the ending values are represented by the two rectangles). However, while looking at variation of the total spread, the second sample shows a better situation as the total spread is lower and it increases more gradually that the first one.   
Thus, the evolution of the two measures should be considered together in order to precisely evaluate the financial risk with respect to its distribution among countries and to the amount of it over a given period.  

To quantify the total credit spread paid the whole sample we proceed as follows. 
Let denote $b_0  = x^c(0) = i$ and define $\forall c \in \mathcal{C}$, $\forall  t \in [\tau_l, \tau_{l+1}-1] $ 
\begin{equation}
\begin{aligned}
\,^{[1:l]}P_{ij}^{(t)} &= \P\left(X(t) = j | X(0) = i\right),
\end{aligned}
\end{equation}
the transition probability from state $i$ to state $j$ in any  given sub-periods. According to  (\ref{edt6}), it is given by 
\begin{equation}\label{ciserve}
\begin{aligned}
\,^{[1:l]}P_{ij}^{(t)} &= 
\sum_{(b_1, \dots, b_l)\in E^l} \prod_{d=0}^{l-1} \,^{(d)}P_{b_d, b_{d+1}}^{(\tau_{d+1}-\tau_d)} \cdot \,^{(l)}P_{b_{d+1},j}^{(t - \tau_{d+1})}. 
\end{aligned}
\end{equation}
\begin{Def}\label{tc}
The expected total credit spread paid by country $c$ given its rating assignment, is defined as 
\begin{equation}\label{etc}
V_i(t) = \E[TC^c(0,t)]	 := \E\left[\sum_{s=1}^t \sum_{j=1}^D \mathbbm{1}_{\{X^c(t) = j | X^c(0) = i\}} S^c(s)\right] .
\end{equation}
\end{Def}
\begin{proposition}\label{p2}
Under assumptions \textbf{A1 - A4} the expected total credit spread paid by country $c$ is given by:
\begin{equation}\label{v1}
V_i(t)	= V_i(t-1) + \sum_{j=1}^D \,^{[1:l]}P_{ij}^{(t)} \cdot \int_0^{+\infty} \overline{F}_j(y)dy,
\end{equation}
\end{proposition}
where $\overline{F}_j(y)$ is the survival function.
\begin{proof}\label{proof1}
$V_i(t)$ can be found recursively. Thus, we can decompose (\ref{etc}) in the following way: 
\begin{equation}\label{v2}
V_i(t) = \E\left[\sum_{s=1}^{t-1} \sum_{j=1}^D \mathbbm{1}_{\{X^c(t) = j | X^c(0) = i\}} S^c(s)\right] +  \E\left[ \sum_{j=1}^D \mathbbm{1}_{\{X^c(t) = j | X^c(0) = i\}} S^c(t)\right] ,
\end{equation}
Under assumption $\textbf{A3}$ and by Definition \ref{tc}, Equation (\ref{v2}) can be rewritten as  
\begin{equation}
\begin{aligned}
V_i(t) =
 &V_i(t-1) +  \E\left[ \sum_{j=1}^D \mathbbm{1}_{\{X^c(t) = j | X^c(0) = i\}} W_j\right]  \\
 &=V_i(t-1) + \sum_{j=1}^D \E\left[ \mathbbm{1}_{\{X^c(t) = j | X^c(0) = i\}}\right] \cdot \E[W_j].
 \end{aligned}
\end{equation}
 Equation (\ref{v1}) holds once we observe that $\sum_{j=1}^D \E\left[ \mathbbm{1}_{\{X^c(t) = j | X^c(0) = i\}}\right]$ is equal to Formula (\ref{ciserve}) and that $E[W_j]=\int_0^{+\infty} \overline{F}_j(y)dy$.
\end{proof}
\begin{remark}
Once the expected credit spread paid by each country has been computed, the expected total credit spread is given by: 

\begin{equation}\label{V}
V(t)=\E[TC(0,t)] = \sum_{j=1}^D n_j(0)\cdot V_j(t),
\end{equation}
where  the  $n_j(0)$,  denotes the initial number of countries allocated in rating class $j$.
\end{remark}
Thus, one could understand if the financial risk is considerable or not, while the measure of entropy suggests a measure of inequality. 

\subsection{Covariance between countries' total credit spread.}
\label{co:variance}
The covariance between the total spread paid by two countries is a useful indicator allowing to understand if the evolution of the credit spreads of some countries are dependent.
In order to compute this indicator we need the expected value of the product of the total credit spread paid by two countries, i.e. $V_{a_\alpha, a_\beta}^{(\alpha,\beta)}(t)$. 
\begin{Def}\label{v_ab}
The expected value of the product of the total credit spread paid by two countries  is defined 
\begin{equation}
\begin{aligned}
V_{a_\alpha a_\beta}^{(\alpha, \beta)}(t) :=  \E \left[ \left( \sum_{s=1}^t \sum_{j_\alpha=1}^D  \mathbbm{1}_{\{X^{(\alpha)}(s)=j_\alpha|X^{(\alpha)}(0) = a_\alpha\}}S^\alpha(s) \right)  \cdot \left( \sum_{s=1}^t \sum_{j_\beta=1}^D  \mathbbm{1}_{\{X^\beta(s)=j_\beta|X^\beta(0) = a_\beta\}}S^\beta(s) \right) \right],
\end{aligned}
\end{equation}
$\forall \alpha, \beta \in \mathcal{C}$ with $X^{\alpha}(0) =a_\alpha$, $X^{\beta}(0) = a_\beta$, $\forall t \in [\tau_l, \tau_{l+1}-1]$.
\end{Def}
\begin{proposition}\label{p3}
Under Assumptions \textbf{A1-A4}, that $V_{a_\alpha a_\beta}^{(\alpha, \beta)}(t)$ is given by
\begin{equation}\label{vab}
\begin{aligned}
V_{a_\alpha a_\beta}^{(\alpha, \beta)}(t)&= V_{a_\alpha a_\beta}^{(\alpha, \beta)}(t-1) 
+ V_{a_\alpha}(t-1) \left( \sum_{j_\beta =1}^D \,^{[1:l]}P_{a_\beta, j_\beta}^{(t)}\cdot \int_0^{+\infty}\bar{F}_{j_\beta}(x)dx \right) \\
& + V_{a_\beta}(t-1) \left( \sum_{j_\alpha =1}^D \,^{[1:l]}P_{a_\alpha, j_\alpha}^{(t)}\cdot \int_0^{+\infty}\bar{F}_{j_\alpha}(x)dx \right) \\
& +  \sum_{j_\alpha =1}^D  \sum_{j_\beta =1}^D \,^{[1:l]}P_{a_\alpha, j_\alpha}^{(t)} \,^{[1:l]}P_{a_\beta, j_\beta}^{(t)}\cdot \int_0^{+\infty} \int_0^{+\infty} z_1\cdot z_2 \cdot f^{(\alpha, \beta)}_{j_\alpha, j_\beta}(z_1,z_2) dz_1dz_2.
\end{aligned}
\end{equation}
\end{proposition}
\begin{proof}
Similarly to the in proof in section \ref{proof1}, we can compute the expected value recursively, so that 
\begin{equation}\label{vab1}
\begin{aligned}
V_{a_\alpha a_\beta}^{(\alpha, \beta)}(t) = & \E \bigg[ \bigg( \sum_{s=1}^{t-1} \sum_{j_\alpha=1}^D  \mathbbm{1}_{\{X^{(\alpha)}(s)=j_\alpha|X^{(\alpha)}(0) = a_\alpha\}}S^\alpha(s)  +  \sum_{j_\alpha=1}^D  \mathbbm{1}_{\{X^{(\alpha)}(t)=j_\alpha|X^{(\alpha)}(0) = a_\alpha\}}S^\alpha(t)\bigg) \\
 &\cdot \bigg( \sum_{s=1}^{t-1} \sum_{j_\beta=1}^D  \mathbbm{1}_{\{X^\beta(s)=j_\beta|X^\beta(0) = a_\beta\}}S^\beta(s) + \sum_{j_\beta=1}^D  \mathbbm{1}_{\{X^\beta(t)=j_\beta|X^\beta(0) = a_\beta\}}S^\beta(t)\bigg) \bigg].\\
\end{aligned}
\end{equation}
Multiplying member by member we obtain:
\begin{equation}\label{vab2}
\begin{aligned}
& \E \bigg[ \bigg( \sum_{s=1}^{t-1} \sum_{j_\alpha=1}^D  \mathbbm{1}_{\{X^{(\alpha)}(s)=j_\alpha|X^{(\alpha)}(0) = a_\alpha\}}S^\alpha(s) \bigg) \cdot \bigg( \sum_{s=1}^{t-1} \sum_{j_\beta=1}^D  \mathbbm{1}_{\{X^\beta(s)=j_\beta|X^\beta(0) = a_\beta\}}S^\beta(s) \bigg) \bigg]\\
&+ \, \E \bigg[ \bigg( \sum_{s=1}^{t-1} \sum_{j_\alpha=1}^D  \mathbbm{1}_{\{X^{(\alpha)}(s)=j_\alpha|X^{(\alpha)}(0) = a_\alpha\}}S^\alpha(s) \bigg) \cdot \bigg( \sum_{j_\beta=1}^D  \mathbbm{1}_{\{X^\beta(t)=j_\beta|X^\beta(0) = a_\beta\}}S^\beta(t)\bigg) \bigg]\\
&+  \,  \E \bigg[ \bigg(\sum_{s=1}^{t-1} \sum_{j_\beta=1}^D  \mathbbm{1}_{\{X^\beta(s)=j_\beta|X^\beta(0) = a_\beta\}}S^\beta(s)\bigg) \cdot \bigg(\sum_{j_\alpha=1}^D  \mathbbm{1}_{\{X^{(\alpha)}(t)=j_\alpha|X^{(\alpha)}(0) = a_\alpha\}}S^\alpha(t)\bigg)\bigg] \\
&+  \E \bigg[ \bigg(\sum_{j_\alpha=1}^D  \mathbbm{1}_{\{X^{(\alpha)}(t)=j_\alpha|X^{(\alpha)}(0) = a_\alpha\}}S^\alpha(t)\bigg) \cdot \bigg( \sum_{j_\beta=1}^D  \mathbbm{1}_{\{X^\beta(t)=j_\beta|X^\beta(0) = a_\beta\}}S^\beta(t)\bigg) \bigg]. 
\end{aligned}
\end{equation}
The first addendum of (\ref{vab2}), by Definition \ref{v_ab}, coincides with the expected value of the product of the total credit spread paid by countries $\alpha$ and $\beta$ at time $t - 1 $, i.e., $ V_{a_\alpha a_\beta}^{(\alpha, \beta)}(t-1)$. 
For the second addendum we have: 
\begin{equation}\footnotesize\label{vab3}
\begin{aligned}
& \E \left[\E\left[ \left( \sum_{s=1}^{t-1} \sum_{j_\alpha=1}^D  \mathbbm{1}_{\{X^{(\alpha)}(s)=j_\alpha|X^{(\alpha)}(0) = a_\alpha\}}S^\alpha(s) \right) \cdot \left( \sum_{j_\beta=1}^D  \mathbbm{1}_{\{X^\beta(t)=j_\beta|X^\beta(0) = a_\beta\}}S^\beta(t) \right) \Bigg{|} \sigma \left(X^{(\alpha)}(z), S^{(\alpha)}(z), X^{(\beta)}(z), z \leq t-1 \right)\right] \right]\\
& = \E\left[ \sum_{s=1}^{t-1} \sum_{j_\alpha=1}^D  \mathbbm{1}_{\{X^{(\alpha)}(s)=j_\alpha|X^{(\alpha)}(0) = a_\alpha\}}S^\alpha(s)\cdot \E\left[\sum_{j_\beta=1}^D  \mathbbm{1}_{\{X^\beta(t)=j_\beta|X^\beta(0) = a_\beta\}}S^\beta(t)   \Bigg{|}  \sigma \left(X^{(\alpha)}(z), R^{(\alpha)}(z), X^{(\beta)}(z), z \leq t-1 \right)     \right] \right].
\end{aligned}
\end{equation}
Now,
 under assumption \textbf{A1} (i.i.d. of $X^{(\alpha)}$ and $X^\beta$) we have 
\begin{equation}\footnotesize\label{vab4}
\begin{aligned}
\E\left[\sum_{j_\beta=1}^D  \mathbbm{1}_{\{X^\beta(t)=j_\beta|X^\beta(0) = a_\beta\}}S^\beta(t)   \Bigg{|}  \sigma \left(X^{(\alpha)}(z), S^{(\alpha)}(z), X^{(\beta)}(z), z \leq t-1 \right)     \right] = \sum_{j_\beta=1}^D \E\left[ \mathbbm{1}_{\{X^\beta(t)=j_\beta|X^\beta(0) = a_\beta\}}W_{j_\beta} \bigg{|} X^{(\beta)}(t-1)\right].
\end{aligned}
\end{equation}
Under assumption \textbf{A3} and because of the Markov property, (\ref{vab4}) can be rewritten as 
\begin{equation}
\sum_{j_\beta=1}^D  \E\left[  \mathbbm{1}_{\{X^\beta(t)=j_\beta|X^\beta(t-1)\}} \right] \cdot \E\left[W_{j_\beta}\right] = \sum_{j_\beta=1}^D \,^{[1:l]}P^{(t)}_{X^{(\beta)}(t-1), j_\beta}\cdot \int_0^{ + \infty} \overline{F}_{j_\beta}(y)dy.
\end{equation}
Thus, (\ref{vab3}) becomes:
\begin{equation}
\E\left[ \sum_{s=1}^{t-1} \sum_{j_\alpha=1}^D  \mathbbm{1}_{\{X^{(\alpha)}(s)=j_\alpha|X^{(\alpha)}(0) = a_\alpha\}}S^\alpha(s) \cdot \,^{[1:l]}P^{(t)}_{X^{(\beta)}(t-1), j_\beta}\cdot \int_0^{ + \infty} \overline{F}_{j_\beta}(y)dy \right],
\end{equation}
which, for \textbf{A1}, can be rewritten as 
\begin{equation}
\begin{aligned}
&\E\left[ \sum_{s=1}^{t-1} \sum_{j_\alpha=1}^D  \mathbbm{1}_{\{X^{(\alpha)}(s)=j_\alpha|X^{(\alpha)}(0) = a_\alpha\}}W_{X^{\alpha}(s)} \right] \cdot \E\left[ \,^{[1:l]}P^{(t)}_{X^{(\beta)}(t-1), j_\beta}\cdot \int_0^{ + \infty} \overline{F}_{j_\beta}(y)dy \right]\\
& =V_{a_\alpha}(t-1) \cdot \left[\sum_{j_\beta=1}^D \,^{[1:l]}P^{(t)}_{a_\beta, j_\beta} \cdot \int_0^{+ \infty} \overline{F}_{j_\beta}(y)dy \right].
\end{aligned}
\end{equation}
which gives us the second addendum of (\ref{vab2}).
The third addendum of  (\ref{vab2}) is  symmetric with respect to the second one, we only need to exchange $\alpha$ and $\beta$. Therefore:
\begin{equation}
V_{a_\beta}(t-1) \cdot \left[\sum_{j_\alpha=1}^D \,^{[1:l]}P^{(t)}_{a_\alpha, j_\alpha} \cdot \int_0^{+ \infty} \overline{F}_{j_\alpha}(y)dy \right].
\end{equation}
Regarding the last addendum of (\ref{vab2}) we have: 
\begin{equation}\small
\begin{aligned}
&\E \bigg[ \bigg(\sum_{j_\alpha=1}^D  \mathbbm{1}_{\{X^{(\alpha)}(t)=j_\alpha|X^{(\alpha)}(0) = a_\alpha\}}W_{X^{(\alpha)}(t)}\bigg) \cdot \bigg( \sum_{j_\beta=1}^D  \mathbbm{1}_{\{X^\beta(t)=j_\beta|X^\beta(0) = a_\beta\}}W_{X^\beta(t)}\bigg) \bigg]\\
&= \E \left[\E\left[ \left(\sum_{j_\alpha=1}^D  \mathbbm{1}_{\{X^{(\alpha)}(t)=j_\alpha|X^{(\alpha)}(0) = a_\alpha\}}W_{X^{(\alpha)}(t)}\bigg) \cdot \bigg( \sum_{j_\beta=1}^D  \mathbbm{1}_{\{X^\beta(t)=j_\beta|X^\beta(0) = a_\beta\}}W_{X^\beta(t)}\right) \bigg{|} X^{(\alpha)}(t), X^{(\beta)}(t)\right] \right].\\
\end{aligned}
\end{equation}
Under assumption \textbf{A4}, 
from 
\begin{equation}\small
= \E\left[ \left(\sum_{j_\alpha=1}^D  \mathbbm{1}_{\{X^{(\alpha)}(t)=j_\alpha|X^{(\alpha)}(0) = a_\alpha\}}W_{X^{(\alpha)}(t)}\bigg) \cdot \bigg( \sum_{j_\beta=1}^D  \mathbbm{1}_{\{X^\beta(t)=j_\beta|X^\beta(0) = a_\beta\}}W_{X^\beta(t)}\right) \bigg{|} X^{(\alpha)}(t), X^{(\beta)}(t)\right],\\
\end{equation}
we obtain 
\begin{equation}
\begin{aligned}
&\sum_{j_\alpha=1}^D  \sum_{j_\beta=1}^D \mathbbm{1}_{\{X^{(\alpha)}(t)=j_\alpha|X^{(\alpha)}(0) = a_\alpha\}} \mathbbm{1}_{\{X^{(\beta)}(t)=j_\beta|X^\beta(0) = a_\beta\}}\cdot \E\left[W_{X^{(\alpha)}(t)} \cdot W_{X^{(\beta)}(t)} \right]\\
&= \sum_{j_\alpha=1}^D  \sum_{j_\beta=1}^D \mathbbm{1}_{\{X^{(\alpha)}(t)=j_\alpha, X^{(\beta)}(t) = j_\beta\}}\cdot \int_{0}^{+ \infty}\int_{0}^{+ \infty} f_{X^{(\alpha)}(t), X^{(\beta)}(t)}^{(\alpha, \beta)}(z_1, z_2) \cdot z_1\cdot z_2 dz_1 dz_2,
\end{aligned}
\end{equation}
where 
\begin{equation}
f_{X^{(\alpha)}(t), X^{(\beta)}(t)}^{(\alpha, \beta)}(z_1, z_2) = \frac{\partial^2F_{X^{(\alpha)}(t), X^{(\beta)}(t)}(z_1, z_2)}{\partial z_1 \partial z_2}.
\end{equation}
Thus, the last addendum of (\ref{vab2}) will be given by 
\begin{equation}
\E\left[\sum_{j_\alpha=1}^D  \sum_{j_\beta=1}^D \mathbbm{1}_{\{X^{(\alpha)}(t)=j_\alpha, X^{(\beta)}(t) = j_\beta\}} \cdot \int_{0}^{+ \infty}\int_{0}^{+ \infty} f_{X^{(\alpha)}(t), X^{(\beta)}(t)}^{(\alpha, \beta)}(z_1, z_2) \cdot z_1\cdot z_2 dz_1 dz_2 \right],
\end{equation}
which leads to 
\begin{equation}
= \sum_{j_\alpha=1}^D  \sum_{j_\beta=1}^D  \,^{[1:l]}P_{a_\alpha, j_\alpha}^{(t)} \,^{[1:l]}P_{a_\beta, j_\beta}^{(t)}\cdot \int_0^{+\infty} \int_0^{+\infty} z_1\cdot z_2 \cdot f^{(\alpha, \beta)}_{j_\alpha, j_\beta}(z_1,z_2) dz_1dz_2.
\end{equation}
\end{proof}

\begin{remark}
 Propositions \ref{p2}-\ref{p3} allow to compute the covariance between the total credit spread paid by  two countries in the following way: 
\begin{equation}
\sigma^{(\alpha, \beta)}_{j_\alpha, j_\beta}(t)= Cov\left(TC^{(\alpha)}(0,t), TC^{(\beta)}(0,t) \right) = V_{a_\alpha a_\beta}^{(\alpha, \beta)}(t)  -  V_{a_\alpha}(t) \cdot  V_{a_\beta}(t) .
\end{equation}

\end{remark}

\section{Results and Discussion}
\label{results}
In this section we show the empirical results we have obtained by applying the whole methodology to the data presented in Section \ref{data:analysis}. After a brief introduction of the change-point detection algorithm, we move to the description and interpretation of the results. In particular, for seek of synthesis, not all results are shown for all agencies. However, the results are available under request.

\subsection{Estimation of the piecewise homogeneous Markov chain (PHMC) and the marginal distributions of the credit spread}
\label{parameter:estimation}
The application starts from the detection of any breaks in the rating dynamics. This is done in order to make the simulations more precise and accurate. In fact, if there are some change-points, it is convenient using the dynamics for the rating process from the last change-point onward.
We rely on  change-point detection theory. In particular we use the off-line detection algorithm proprosed by Polansky \cite{detect}. 

Given that we suppose the existence of some change-points but their position and their number are unknown, we have to proceed as follows. We start by supposing that there is one and more change-points; we apply the  likelihood theory in order to detect the exact position of the change-points. Finally, we find the optimal number of change points (whose position is previously found) by means of the Bayesian information criterion (BIC). The BIC allows to choose the best model to fit the observed data, by balancing the goodness of fit of the model  with the number of parameters required.

Therefore, when the change-points  $ \tau_1, \dots, \tau_k$ are unknown parameters but the number $k$ of change-points is known, the 
the maximum likelihood estimator (MLE) of the change-points, i.e. ,   $\hat{\tau}_1, \dots , \hat{\tau}_k$, is estimated by maximizing the likelihood function. 
\begin{equation}
(\hat{\tau}_1, \dots , \hat{\tau}_k)  = \underset{\tau_1<\dots < \tau_k \in \{1, \dots , n-1  \}}{\arg \max} \Big{\{} \sum_{m=0}^k L(\tau_m, \tau_{m+1}) \Big{\}},
\end{equation}
where $\sum_{m=0}^k L(\tau_m, \tau_{m+1}) $ is the  observed likelihood function conditional on the changing points (details about the computation of the likelihood function can be found in \cite{detect}).
To understand if the breaks we have found result on abrupt changes in the rating process we test 
$H_0 : \,^{(0)}\textbf{P} = \dots = \,^{(k)}\textbf{P}$ against $H_1 : \,^{(0)}\textbf{P} \neq \dots \neq \,^{(k)}\textbf{P}$.

The test statistic $\Lambda$ applied to the sequences of observations $x^c(1) \dots , x^c(s)$, $\forall c \in \mathcal{C}$, is computed as follows: 
\begin{equation}
\Lambda = \sum_{m=1}^k \sum_{c \in \mathcal{C}} \sum_{i=1}^s \frac{p_{ij}^{(m)}(x_i^c, x_{i+1}^c)}{p_{ij}^{(m-1)}(x_i^c, x_{i+1}^c)},
\end{equation}
which  holds to the following Equation:
\begin{equation}\label{test}
\Lambda= 2 \bigg{(}\sum_{m=1}^k L(\tau_m,\tau_{m+1}) - L(\tau_0, \tau_{k+1})\bigg{)}.
\end{equation}
When the change-points are known, $\Lambda$ asymptotically tends to a $\chi^2$ distribution with $D(D-1)k$ degrees of freedom. On the contrary, when the change-points are unknown,
the critical value of the test is derived using the bootstrap simulation (for further details on this techinque see \cite{detect}). The null hypothesis, with significance level $\alpha$, is rejected if $\Lambda > \Lambda_{|1-\alpha|}$.

Lastly, a model selection is required in order to avoid problem of over-fitting due to a great amount of parameters. Thus, the BIC
allows to choose the best model to fit the observed data, by balancing the goodness of fit with the number of parameters. It is defined as:  
\begin{equation}\label{bic}
BIC(k) = \log(s) \cdot D(D-1)\cdot (k+1)  - 2   \sum_{m=0}^k L \big{(}\hat{\tau}_m, \hat{\tau}_{m+1}\big{)},
\end{equation}
where the first term on the right side of  (\ref{bic}) denotes the number of unknown parameters and the second one the goodness of fit of the model. The optimal number of change-points is detected by minimizing the BIC  previously computed for $k=0, 1, ...K$, with $K$ is arbitrarily chosen.

The change-point detection algorithm is applied to the observed rating data for the three agencies. It is worth noting that, as we are working on a daily scale, the observations are $128\,976$ for each agency. Find out the position of the change-points when $k=2, 3$ creates computational issues, mostly related to the huge amount of time required.
To overcome this problem, we firstly detect the breaks on a monthly scale. Then, we use the results to build a range of time, on a daily scale, within whom the algorithm is carried out. Furthermore, the algorithm is embarrassingly parallel allowing for more speed without loosing informations over the data.
\begin{table}[!ht]
\centering
\begin{tabular}{lllll}
\hline
 & S\&P &  &  &  \\
\hline
$k$ &0&	1&2&3\\
L&-542.61&	-483.51&	-447.98&	-427.46\\
Parameters& 15&	24&	36&48\\
BIC & 1214.1&\textbf{1070.1}&1205.2&	1267.2\\
\hline
&Moody's & &  & \\
\hline
$k$ &0&	1	&2&	3\\
L&	-470.70&415.06&-391.33&	-361.76\\
Parameters&16&	32&	42&	36\\
BIC &	1078.83&	1104.98 &	          1143.42&\textbf{1032.74}\\
\hline
&Fitch&  &&\\
\hline
$k$ &0&	1	&2&	3\\
L&	-530.02&	-496.57&	-469.41&	-472.59\\
Parameters&14&	24&	30&	40\\
BIC &	\textbf{1180.30}&	1199.29 &	1196.51&	1288.75\\
\hline
\end{tabular}
\caption{Results from change-points detection algorithm}\label{b.i.c}
\end{table}

Table \ref{b.i.c} gives results about the number of change-points ($k$), the maximum  likelihood function (L), the number of parameters\footnote{In our application, the number of parameters is equal to  $|G|$
 where $G = \{P_{ij} : P_{ij}>0, i \neq j \}$} and the value of the  Bayesian information criterion (BIC). 
According to our results, the best model for S\&P is that including one change-point, detected on January $12$, $2012$. For Fitch there are no change-points. For Moody's the best model includes three change-points: November 11, 2002; March 30, 2009 and April 29, 2013. These findings are supported by the results of  statistical test $\Lambda$ at a significance level of 0.05 as shown in Table (\ref{ciserve2})
\begin{table}[!ht]
\centering
\begin{tabular}{c|ccc}
&$\Lambda_{0.95}$&  $\Lambda$& p-value\\
\hline
S\&P & 262.659& 447.228&0.0099\\
Moody's & 8.369&217.878&0.0049\\
\end{tabular}
\caption{Results of the statistical test: $\Lambda_{0.95}$ stems from the bootstrap simulation; $\Lambda$ is the statistic computed on the observed data.}\label{ciserve2}
\end{table}
To compute the expected inequality and the other measures we proposed in Section \ref{theil:index}, the transition probability matrix is estimated using data from the last changing point onward. In particular, for Fitch we use the whole observed rating trajectories because no change-points was detected.
\begin{table}[!h]\footnotesize
\begin{tabular}{c|cccccccc}
&$1$&$2$&$3$&$4$&$5$&$6$&$7$&$8$\\
\hline
$1$&$0.99948$    & $5.15e-04$  &	$0$      &	  $0$    &	$0$      &	$0$       &	$0$       &	$0$
\\
$2$&$9.71e-05$   &   $0.99971$ &  $1.94e-04$ &	  $0$    &	$0$      &	$0$       &	$0$       &	$0$
\\
$3$&$0$          &	  $0$      &  $0.99954$  &$4.57e-04$ &	$0$      &	$0$	      & $0$       &	$0$
\\
$4$&   $0$       &    $0$      & $3.77e-04$  & $0.99934$ & $2.83e-04$&	$0$       &	$0$       &	$0$
\\
$5$&   $0$       &	  $0$      &	$0$      & $6.02e-04$&  $0.9994$ &	$0$       &	$0$       &	$0$
\\
$6$&   $0$	     &    $0$      &	$0$      &	  $0$    &   $0$     & $0.99931$  & $6.92e-04$&	$0$
\\
$7$&   $0$       &    $0$	   &    $0$      &	  $0$    &	 $0$     & $0.00208$  & $0.99375$ & $0.00417$	
\\
$8$&   $0$       &	  $0$      &	$0$      &	  $0$    &	 $0$     & $0.01333$  & $0.01333$ & $0.97333$
\\
\end{tabular}
\caption{Transition probability matrix after the change-point ($2012-01-12$)}\label{psep}
\end{table}
Table (\ref{psep}) represents the transition probability matrix resulting from the S\&P data related to the last sub-period. The probability of maintaining the current state is very high and it is decreasing as the credit quality goes down. Furthermore, for rating classes A and B ( i.e. $g_i \in [3,6]$) there is no probability of upgrade. 

We compute the Jafry Schuermann distance $d_{JS}$ (see \cite{JS}) in order to compare the transition probability matrix of the PHMC with that of the traditional HMC.
The index proposed in Jafry and Schuermann  measures the average probability of transition by assessing the mean of all singular values $\lambda_i $ of the product $\tilde{\textbf{P}}' \tilde{\textbf{P}}$, where  $\tilde{\textbf{P}}'$  is the transpose of the mobility matrix  $\tilde{P}$. In particular, $\tilde{\textbf{P}}=\textbf{P} - \textbf{I}$, where $\textbf{I}$ is the identity matrix with the same dimension as the original transition probability matrix $\textbf{P}$:
\begin{equation}\label{svd}
M_{SVD}(\textbf{P})= \frac{\sum_{i=1}^N \sqrt{\lambda_i (\tilde{\textbf{P}'}\tilde{\textbf{P}})}}{N},
\end{equation}
This metric  can be interpreted as a proxy of the average transition probability of a given migration matrix.  

Then, the difference is calculated by comparing the average singular value (obtained in Formula (\ref{svd})) of the three matrices:
\begin{equation}
D_{JS}= M_{SVD}(^{(k)}\textbf{P}) - M_{SVD}(\textbf{P}).
\end{equation}
where $^{(k)}\textbf{P}$ denotes the transition probability matrix of the PHMC and $\textbf{P}$ denotes that of the HMC. 
This metric is very close to zero ($d_{JS}=1.42e-04$ for S\&P and $d_{JS}=3.63e-04$ for Moody's). Thus, there is almost the same average transition probability in both cases, although the test exposed above said us that the PHMC are significantly different.
This measure is also applied to compare the matrices of all agencies, we observe that the value of  the distance ranges between $0.039$ and $0.044$ as described in the  Table (\ref{ciserve3}):
\begin{table}[!ht]
\centering
\begin{tabular}{c|ccc}
&Fitch&Moody's& S\&P	\\
\hline
Fitch&	-& 0.043805&	0.039829\\
Moody's& -0.043805	&-&	-0.0039757\\
S\&P &-0.039829&	0.0039757&-\\
\end{tabular}
\caption{$d_{JS}$ applied to the transition matrices of all agencies.}\label{ciserve3}	
\end{table}

\begin{figure}[!ht]
\begin{minipage}{0.45\textwidth}
\includegraphics[width=0.95\textwidth]{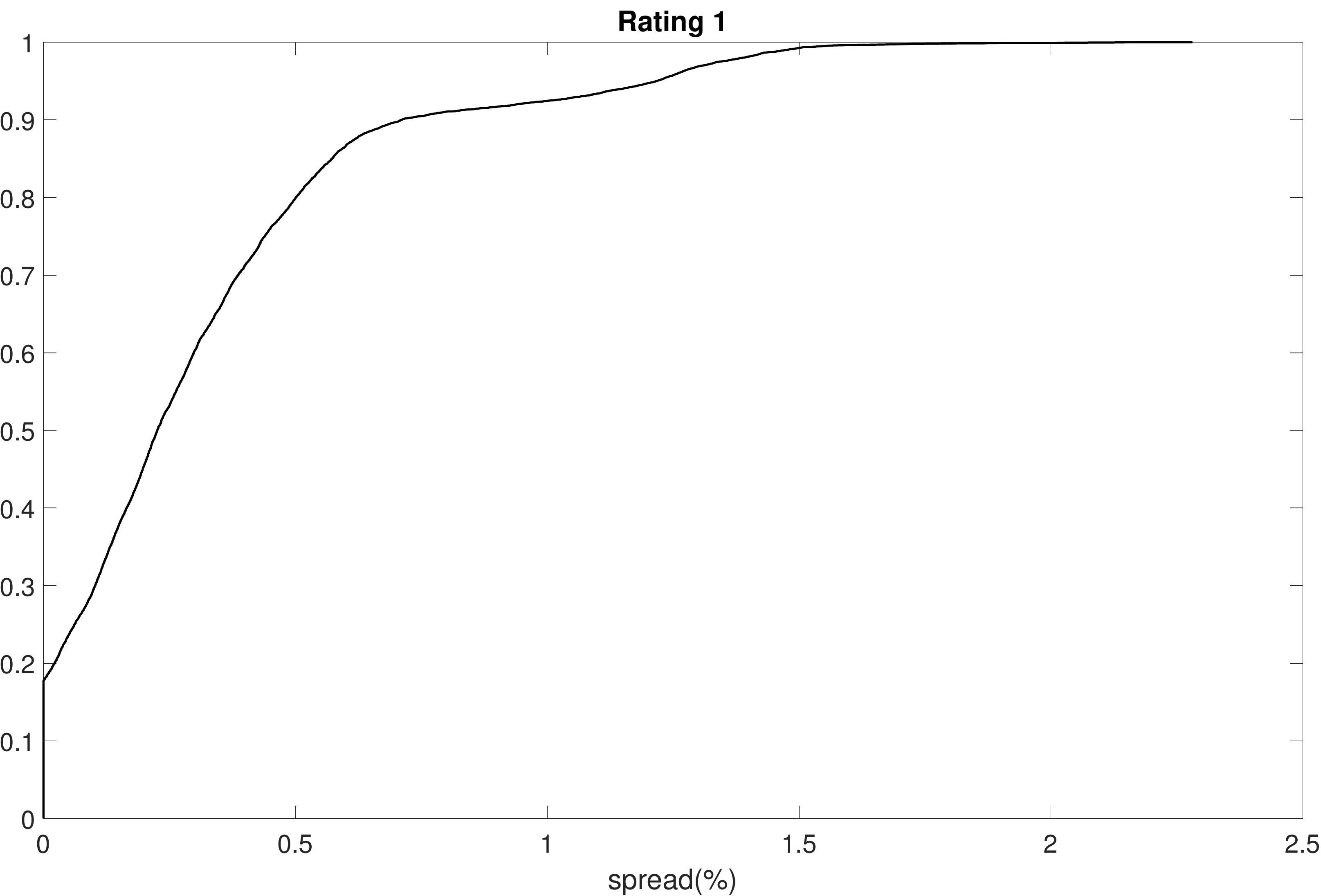}
\end{minipage}
\begin{minipage}{0.45\textwidth}
\includegraphics[width=0.95\textwidth]{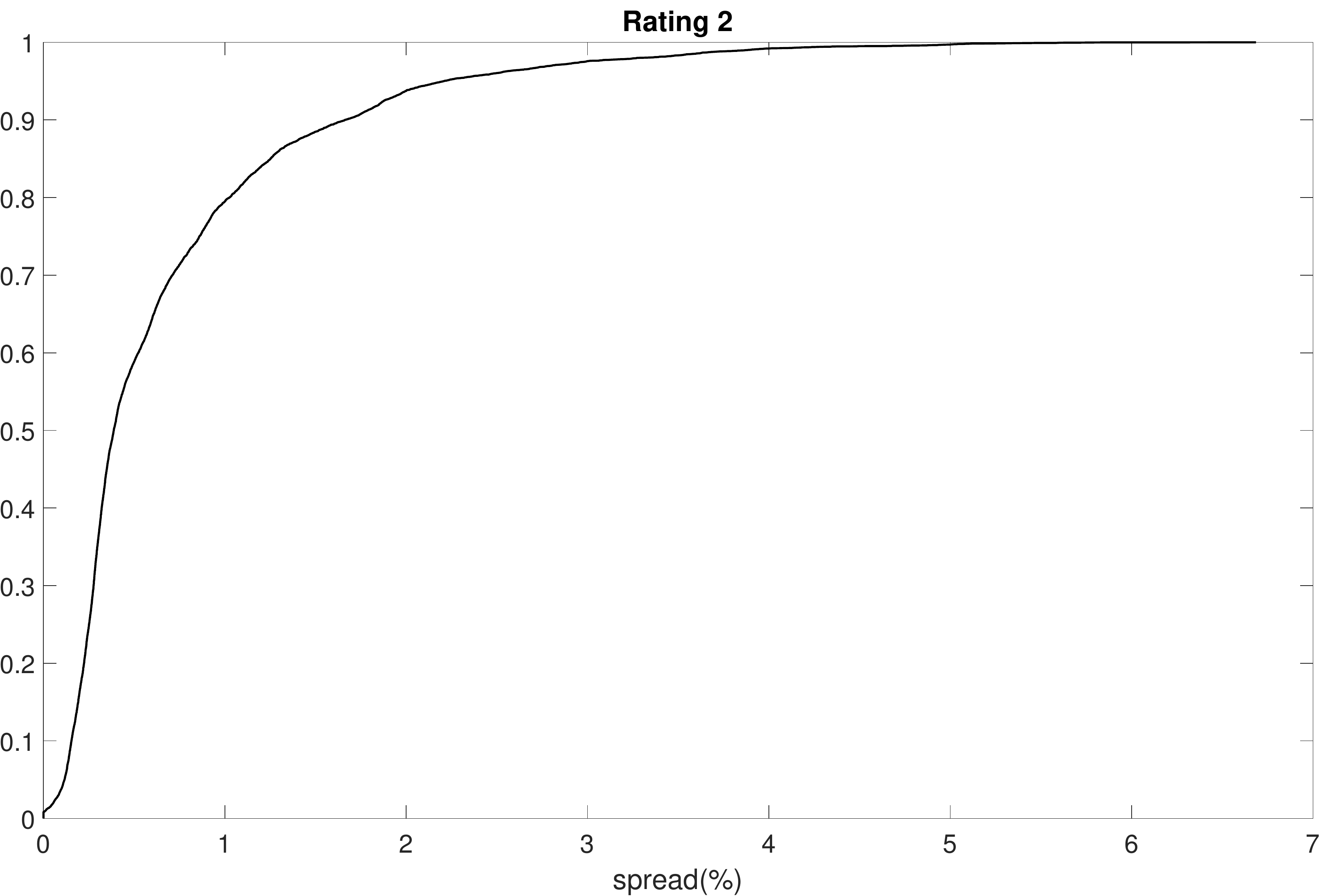}
\end{minipage}\\
\begin{minipage}{0.45\textwidth}
\includegraphics[width=0.95\textwidth]{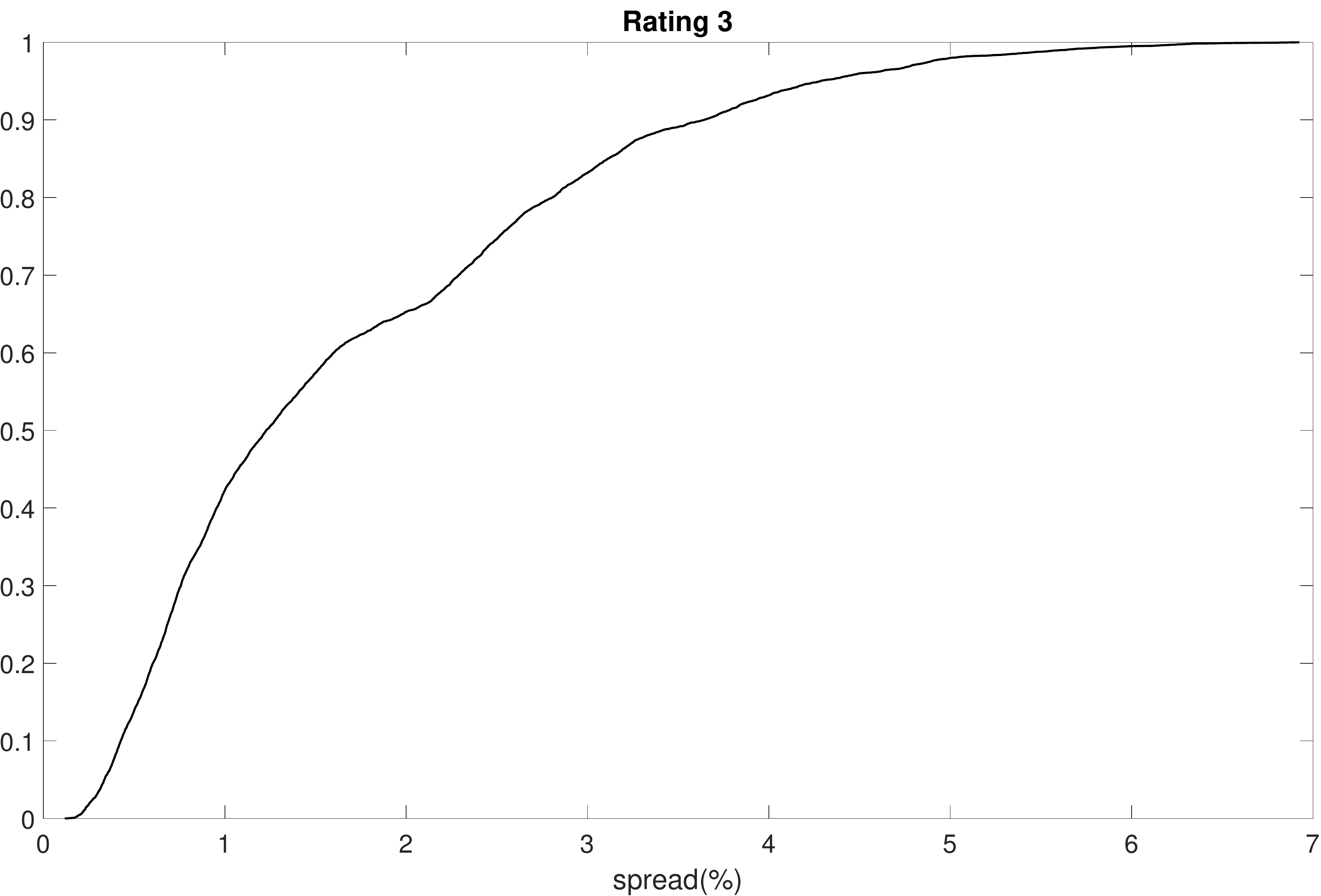}
\end{minipage}
\begin{minipage}{0.45\textwidth}
\includegraphics[width=0.95\textwidth]{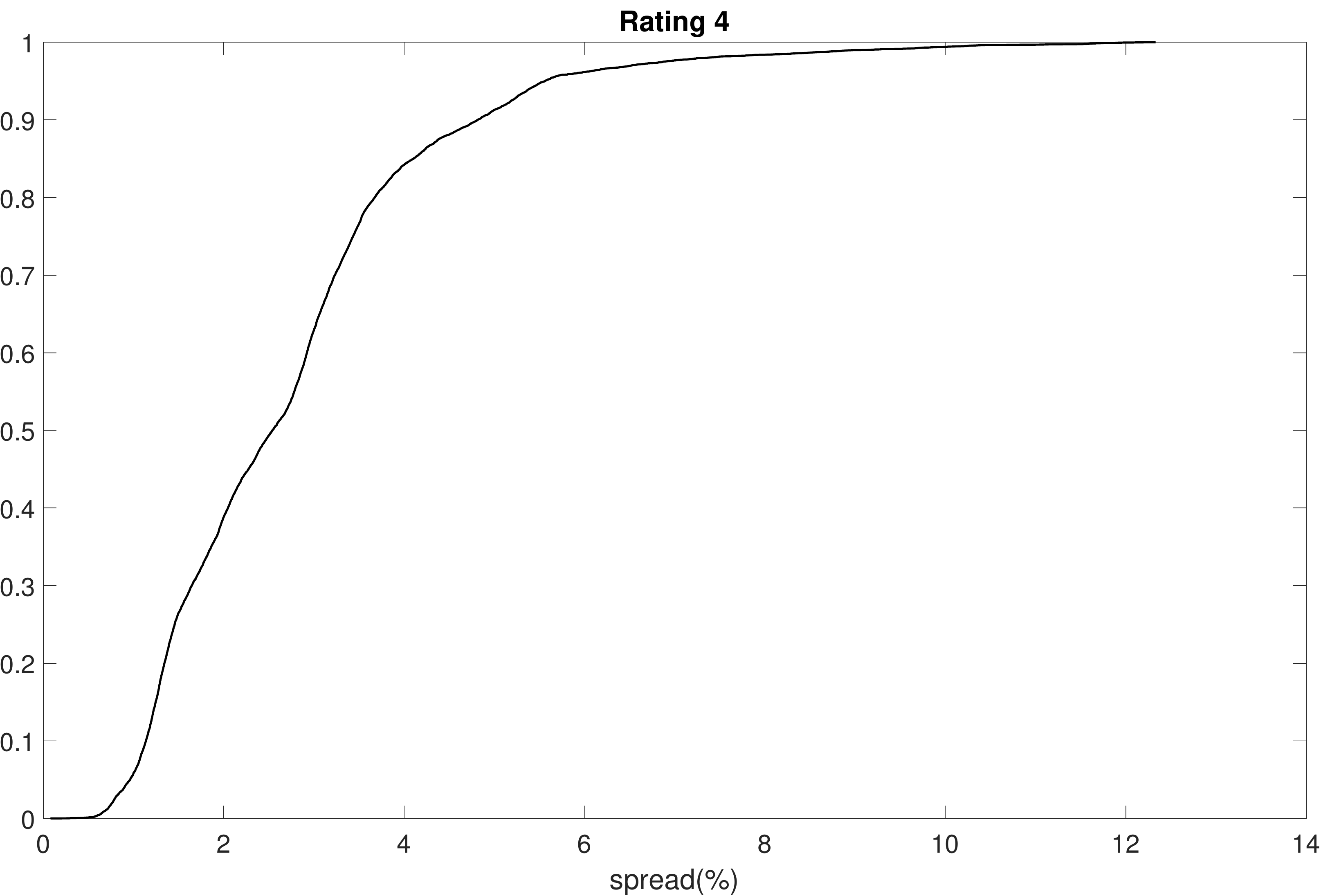}
\end{minipage}\\
\begin{minipage}{0.45\textwidth}
\includegraphics[width=0.95\textwidth]{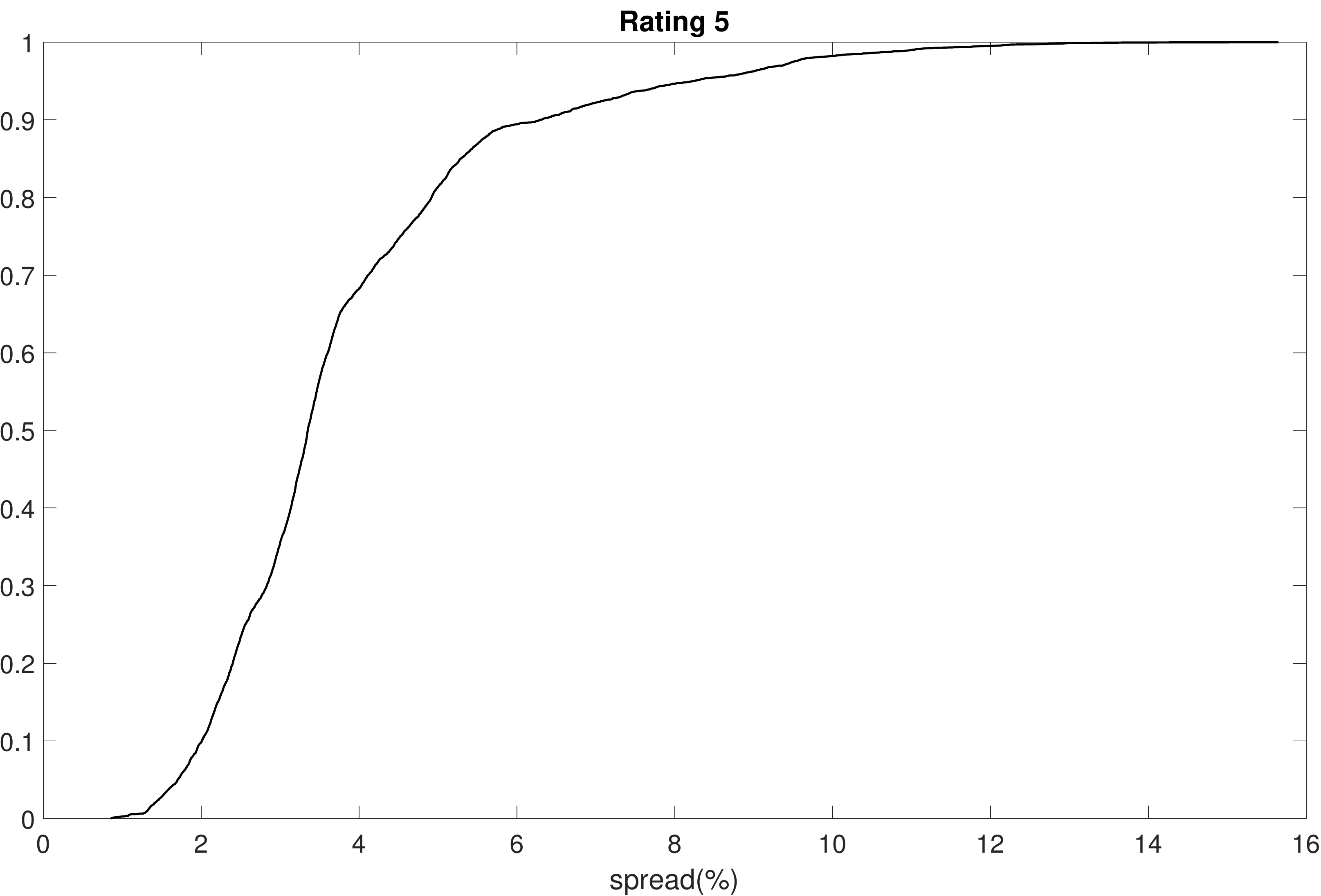}
\end{minipage}
\begin{minipage}{0.45\textwidth}
\includegraphics[width=0.95\textwidth]{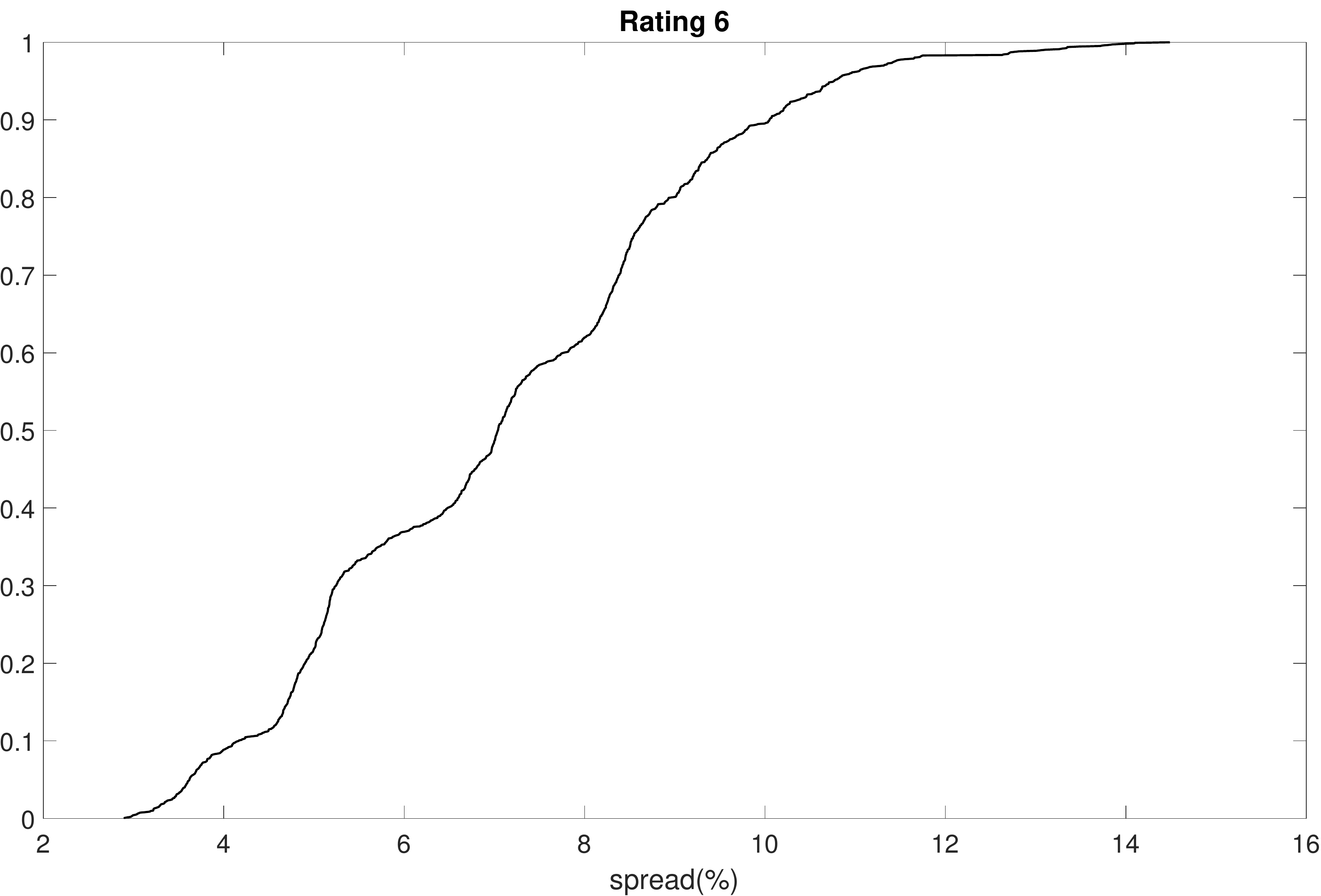}
\end{minipage}\\
\begin{minipage}{0.45\textwidth}
\includegraphics[width=0.95\textwidth]{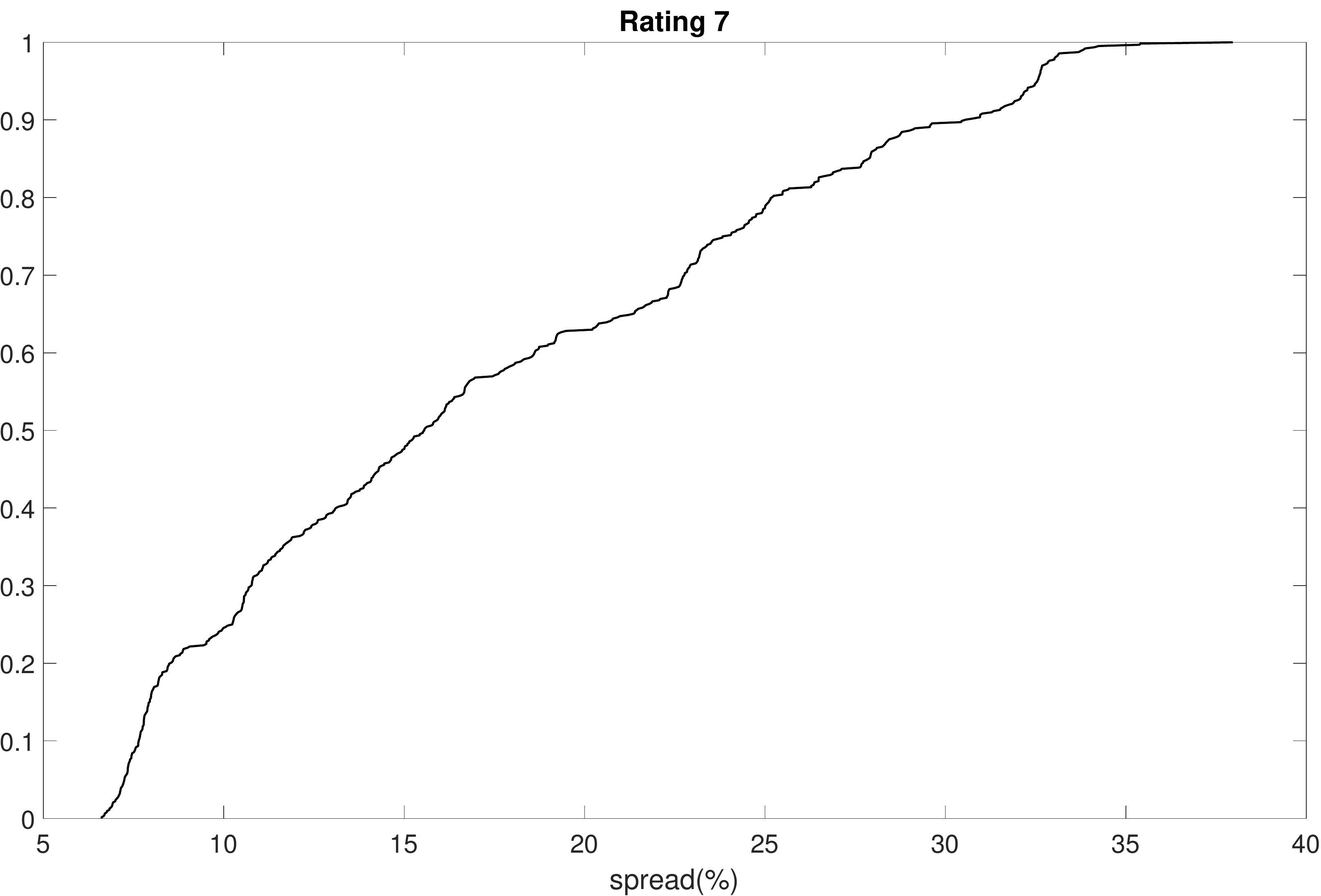}
\end{minipage}
\caption{Empirical distributions for rating class with S\&P data.}\label{cdf}
\end{figure}
Regarding the credit spread, to compute the multivariate distribution we make use of a Copula applied to the  distribution of $N$ countries' credit spread. In particular, according to Assumption \textbf{A3}, the marginal distributions we need are those related to the rating classes. Figure (\ref{cdf}) represents the empirical credit spread distribution for all rating classes, resulting from S\&P data. The credit spread are expressed in basis point values. Each sub-figure shows the empirical c.d.f. starting from rank=1 to rank=7. 
Moreover, some descriptive statistics related to those distributions are shown in Table (\ref{descriptive:statistics}). 
\begin{table}[!ht]
\centering
\begin{tabular}{|ccccccccc|}
\hline
rank& 1&2&3&4&5&6&7&8\\
\hline
sample& 13272	&11896&	11221&	12784&	7376&	1472	&632	&75\\
\hline
Mean& 0,321	&0,696&	1,700&	2,750&	3,834&	7,053&	17,356&	21,029\\
\hline
st.dev.&  0.350 &  0.772 &    1.298&    1.635&    1.954&    2.254&    8.344&    7.451\\
\hline
Skew & 1.743&    2.651&    1.120&    1.731&    1.805&    0.341&   0.466&    1.169\\
\hline
Kurt & 6.056&   11.882&    3.702&    7.914&    6.982&    2.656&    1.969&    3.199\\
\hline
\end{tabular}
\caption{Descriptive statistic of credit spread distribution with S\&P data.}\label{descriptive:statistics}
\end{table}
Not surprisingly, the average value increases as the credit quality gets worst (i.e. from rank=1 to rank=8). Higher risk perceived corresponds to a higher cost of debt and this cost is more variable as the credit quality decreases. Furthermore, the number of observations for the investment rating classes (i.e. rank=1, ..., 4) are larger with respect to the speculative rating classes (rank = 5, ... ,8).

The spread distributions are tested in order to understand if the use of different distributions for the rating classes is justified by a significant dissimilarity between them. The results of the Anova test confirmed that the spread distributions are different,  at a significance level of 0.05.
\begin{table}[!ht]
\centering
\begin{tabular}{c|ccc}
&Fitch&Moody's& S\&P	\\
\hline
F&11934,41&11725.85&17554.65\\
p-value& 0.00&0.00&0.00\\
\end{tabular}
\caption{Anova test within all rating agencies.}
\end{table}

\subsection{Assessing the financial inequality and the total credit spread.}
\label{financial:indicators}
The financial risk inequality is assessed  by means of the expected value of the dynamic Theil entropy as shown in Section \ref{dynamic:inequality}. This requires the set of all possible configurations stemming from 24 countries and 8 rating class (2 629 575). To avoid computational problems due to the huge amount of combinations, the Monte Carlo is carried out starting from modelling rating trajectories according to  $\,^{(k)}P_{ij}$ for each rating agency. In particular, by observing the initial configuration of countries, if a given country $c$ visited the state $i$ at time $t$ we take the c.d.f. of the probability distribution for this state (i.e. $p_{i,\cdot}$). After that, a pseudo-random number $u\in [0,1]$ is generated such that, if 
\begin{equation*}
\sum_{j=1}^k \,p_{i,j}\leq u(t) < \sum_{j=1}^{k+1}\, p_{i,j}, \hspace{0.5cm} \forall i,k \in E,
\end{equation*}
the next rating class visited by country $c$ at time $t+1$ would be equal to $k$. This is done for a horizon period of three years starting from the allocation of countries observed at the end of our datasets. Then, the spread dynamics are also simulated, according to the rating simulated at each time $t$, by extracting the spread from a multivariate distribution. This is done by applying a Gaussian Copula to the empirical distribution of the rating classes shown in the previous section. 
In order to better simulate the future payments we estimate the relative variation of the observed credit spread paid by each country. 
Once the spread dynamics are generated, the inequality $\D$ is computed for each Monte Carlo iteration. This process has been performed for $200$ iterations.  
 Finally, the first  order moment, i.e. $\E\left[T\D\right]$  is computed through the mean of the inequalities over all iterations. 
Figure (\ref{edttc}) shows the evolution of the expected dynamic Theil entropy along with the variation of the Expected total credit spread, i.e. $\Delta V(t)= V(t)-V(t-1)$. The absolute variation of the total credit spread is expressed in percentage value. The initial and ending values are represented by a circle and a rectangle, respectively. 
\begin{figure}[!ht]
\centering
\includegraphics[width=0.7\textwidth]{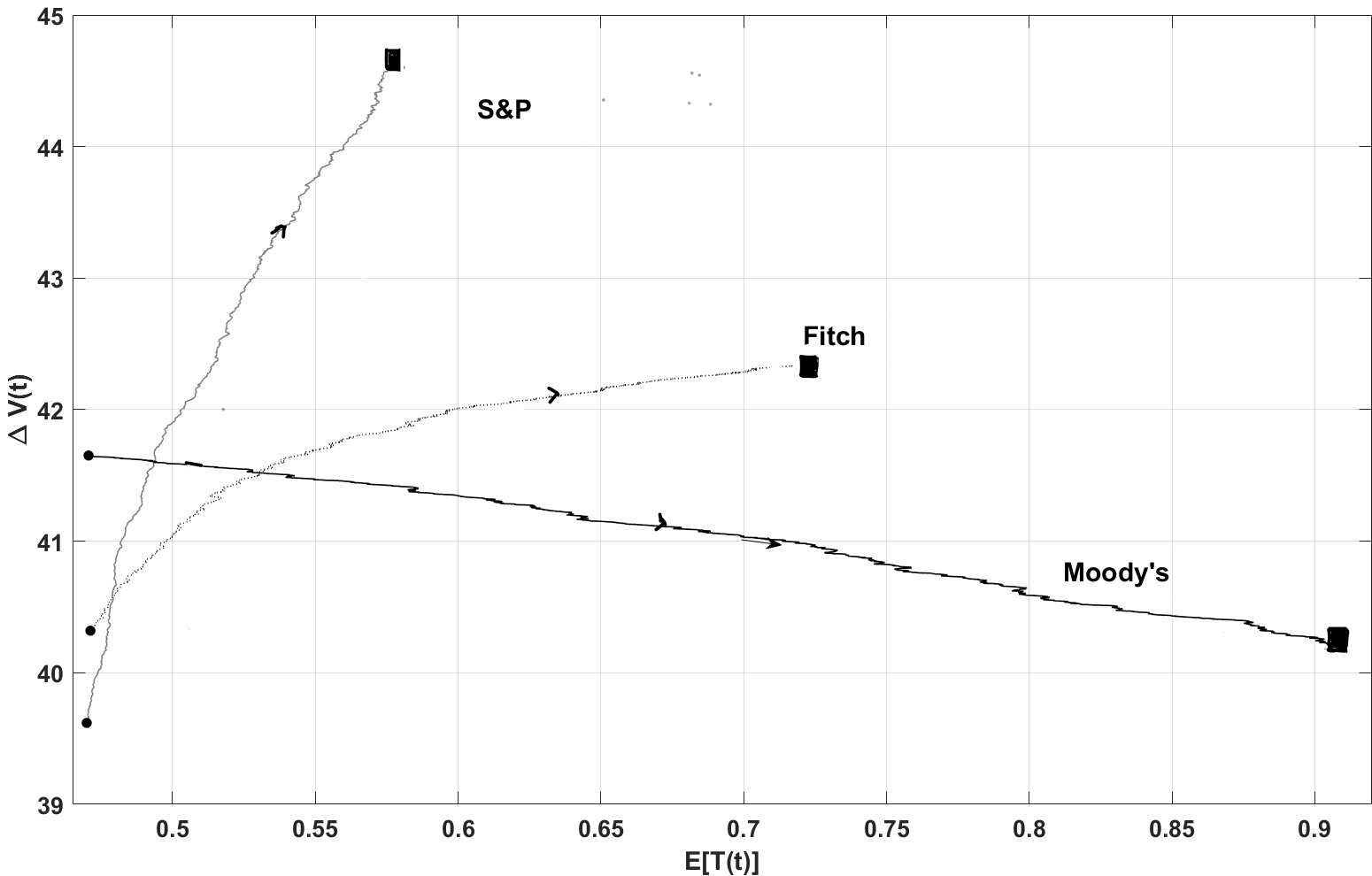}
\caption{Expected inequality and absolute variation of the the total credit spread for all agencies}\label{edttc}
\end{figure}

The expected inequality shows an upward trend for all three agencies. Starting from the same value, $0.47$ (abscissas of circle points), the expected inequality evolves differently among the rating agencies.
In the case 
of Standard \& Poor's it  increases of about the $23\%$ of the initial value ending with value equal to $0.578$ (abscissa of rectangle point of S\&P curve). For Fitch the inequality goes up until $0.721$ (abscissa of rectangle point of Fitch curve). The highest one is that computed for Moody's: the expected inequality, in fact, grows more than the other agencies and, at end of the horizon time, is $0.91$ (abscissa of rectangle point of Moody's curve). These results suggests that, in all cases, the financial risk  would be less equi-distributed after three years. 

On the ordinate axis of the Figure (\ref{edttc}) the variation of the expected total credit spread is shown. It expresses  the speed of growth of the expected total credit spread over time. 
$\Delta V(t)_{S\&P}$ increases over the period ranging between 39.625 (ordinate of circle point of S\&P curve) and 44.638 (ordinate of rectangle point of S\&P curve). In a similar way $\Delta V(t)_{Fitch}$ increases but ranging between 40.331 and 42.340. On the contrary, $\Delta V(t)_{Moody's}$ slowly decreases over time starting from 41.646 up to 40.205. 

The expected total credit spread helps us to quantify the size of the financial risk and to better interpret the evolution of the financial risk. Table (\ref{tabletc}) shows the values of the total credit spread at the beginning or the simulation, after the first two years and at the ed of our simulations. The value are expressed in percentage and for the three agencies.  
\begin{table}[!ht]
\centering
\begin{tabular}{|r|cccc|}
\hline
&          $t=0$       &     $1^{st}$ year   &     $2^{nd}$ year   &    $3^{rd}$ year\\
\hline
Moody's&        $41.648$     &     $15106.994$     &     $30032.948$     &    $44789.140$  \\ 
\hline
Fitch&          $40.325$     &     $14994.514$     &     $30272.184$     &    $45677.702$  \\
\hline
S\&P &          $39.615$     &     $14995.465$     &     $30659.464$     &    $46759.857$  \\ 
\hline 
\end{tabular}
\caption{$V(t)=\E[TC(0,t)]$ for all agencies.}\label{tabletc}
\end{table}
By looking at the Table (\ref{tabletc}), it is evident that the evolution of this indicator is different among rating agencies. 
After the first year Moody's agency  has the highest total payment, followed by S\&P. On the contrary, starting from the second year the total credit spread computed for Moody's  has the smallest value followed by Fitch. 

This can be summarize as follows. In the scenario of S\&P there would be more equi-distribution than the other agencies. However, the amount of the financial risk is higher and it increases faster than the other agencies.Furthermore, with Fitch data $\E\left[\D\right]$ increases more than S\&P but the total credit spread increases at lower rates. These results suggest that also in this case there are some countries paying more than the others but the difference among countries is higher even if the total risk is smaller. Finally, in the scenario simulated using Moody's data $\E\left[\D\right]$   goes up of almost the $90\%$ of its observed value, but the total credit spread increases at decreasing rates meaning that the increasing spread is paid by a subgroup of countries that is smaller than in the other two scenarios. 

The property of decomposition of  the dynamic Theil index allows to investigate the influence of the rating dynamics on the financial risk inequality assessment. 
\begin{figure}[!ht]
\centering
\includegraphics[width=1.1\textwidth]{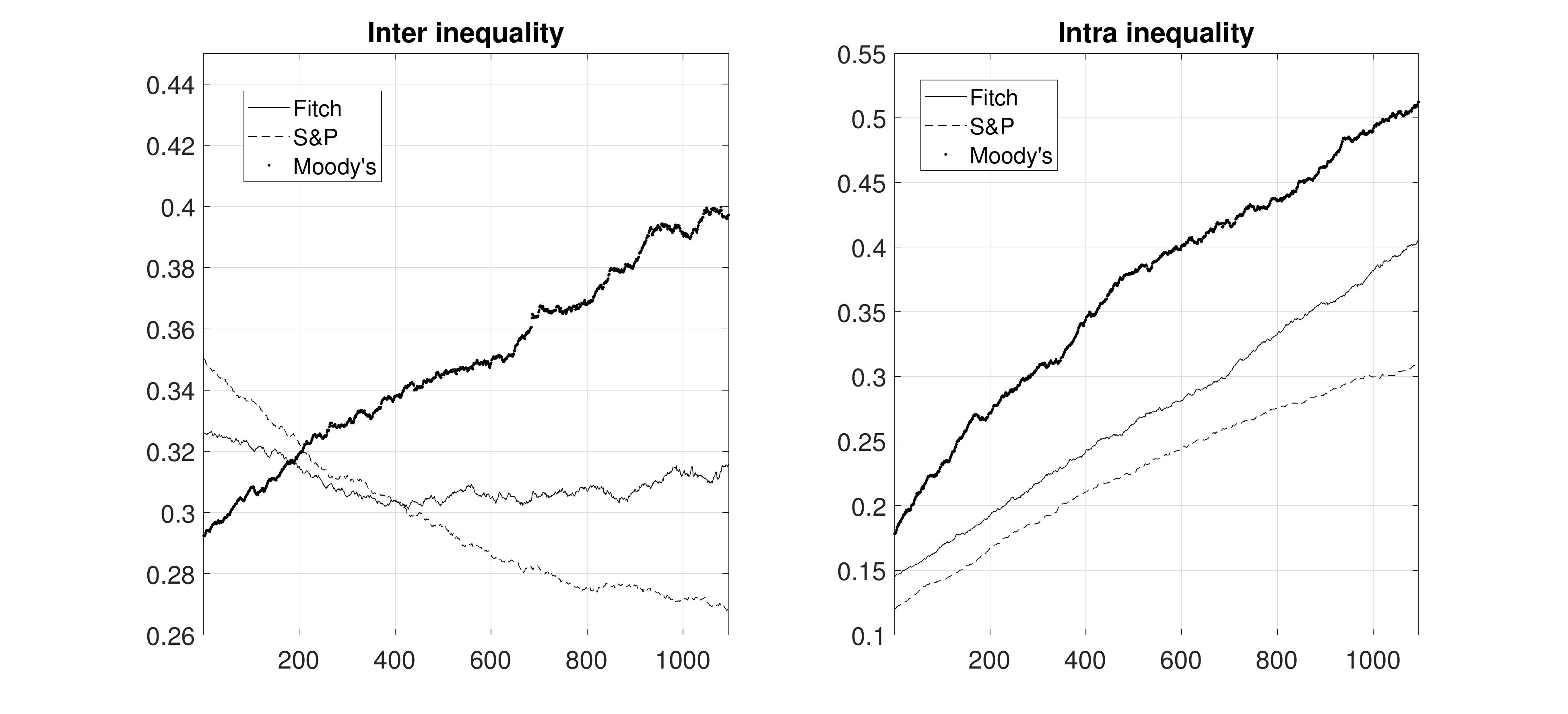}
\caption{Inter  and Intra - class inequality measure for all agencies}\label{inter}
\end{figure}
Figure (\ref{inter}) shows the dynamic inequality measure decomposed into inter - inequality (left panel) and intra - inequality (right panel). The inter - inequality assess the dispersion of the financial risk between rating classes. The intra - inequality is the inequality within rating classes, computed on the conditional share of credit spread knowing the allocation of each country in a given rating class. According to Figure (\ref{inter}) the intra-inequality measure is increasing for all agencies, with Moody's being the one with the largest values, followed by Fitch and S\&P. On the other side, the inter-inequality measure shows different trends. As a matter of fact, for Moody's it is increasing, explaining the strong growth of the financial risk inequality for this agency. On the contrary, for the other agencies the inequality computed between rating classes is decreasing. In particular, for Fitch it falls very slowly until the middle of the simulated period. After it remains almost stable. Thus, the rise of the financial inequality is explained by the evolution of the intra-class inequality. While, in the case of S\&P, the inter-class  inequality decreases over time of about the 23\%. As the intra-class inequality increases, the resulting inequality of the financial risk goes up (as shown in Figure \ref{edttc}). We observe that higher values of the inter-inequality measure imply a strong explanatory power of the rating classes. On the contrary, lower values of the inter-inequality measure denote a relative small influence of the credit rating while explaining the inequality evolution. 
  
\subsection{The covariance between countries' total credit spread and the correlation structure}
The last result concerns the covariance between countries' total credit spread as described in Section \ref{co:variance}. To better interpret the results we compute the coefficient of correlation. We will denote this coefficient as $\rho_{j_\alpha, j_\beta}^{\alpha,\beta}(t)$.  Figures (\ref{corr1}-\ref{corr2}) compare the correlation coefficients estimated on the cumulative observed spread, i.e. $\rho^{\alpha,\beta}$ on the left panel, with $\rho_{j_\alpha, j_\beta}^{\alpha,\beta}(t)$ computed using S\&P data on the right panel. In particular, the value of  $\rho_{j_\alpha, j_\beta}^{\alpha,\beta}(t)$ is taken in the middle of the simulated period, i.e. $t=547$. So that $\rho^{\alpha,\beta}$ is estimated on the total spread paid by all countries computed on the observed data and over the same time-length. 
\begin{figure}[!ht]
\begin{minipage}{0.49\textwidth}
\includegraphics[width=1.1\textwidth]{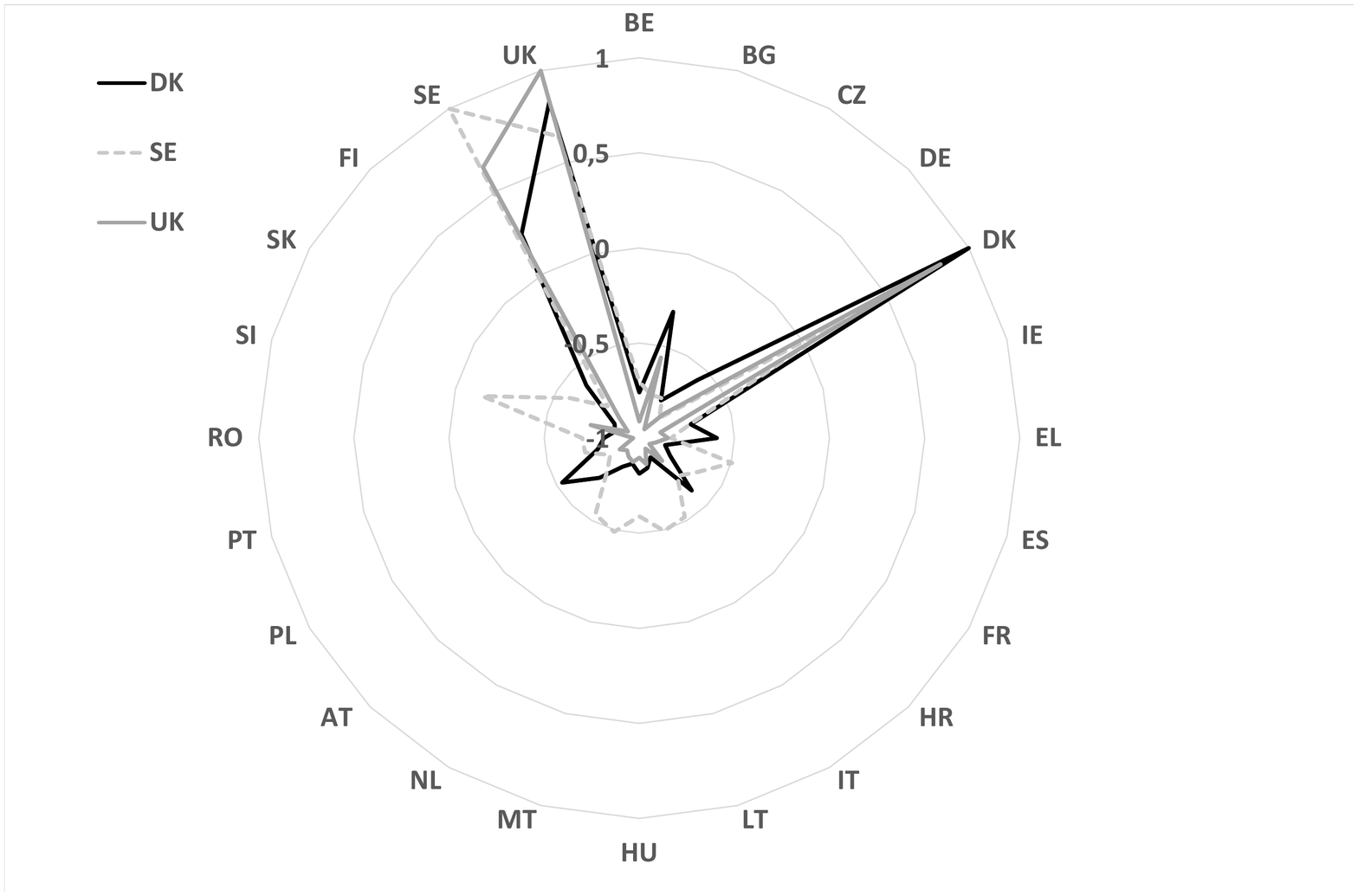}
\end{minipage}
\hspace{0.2cm}
\begin{minipage}{0.49\textwidth}
\includegraphics[width=1.1\textwidth]{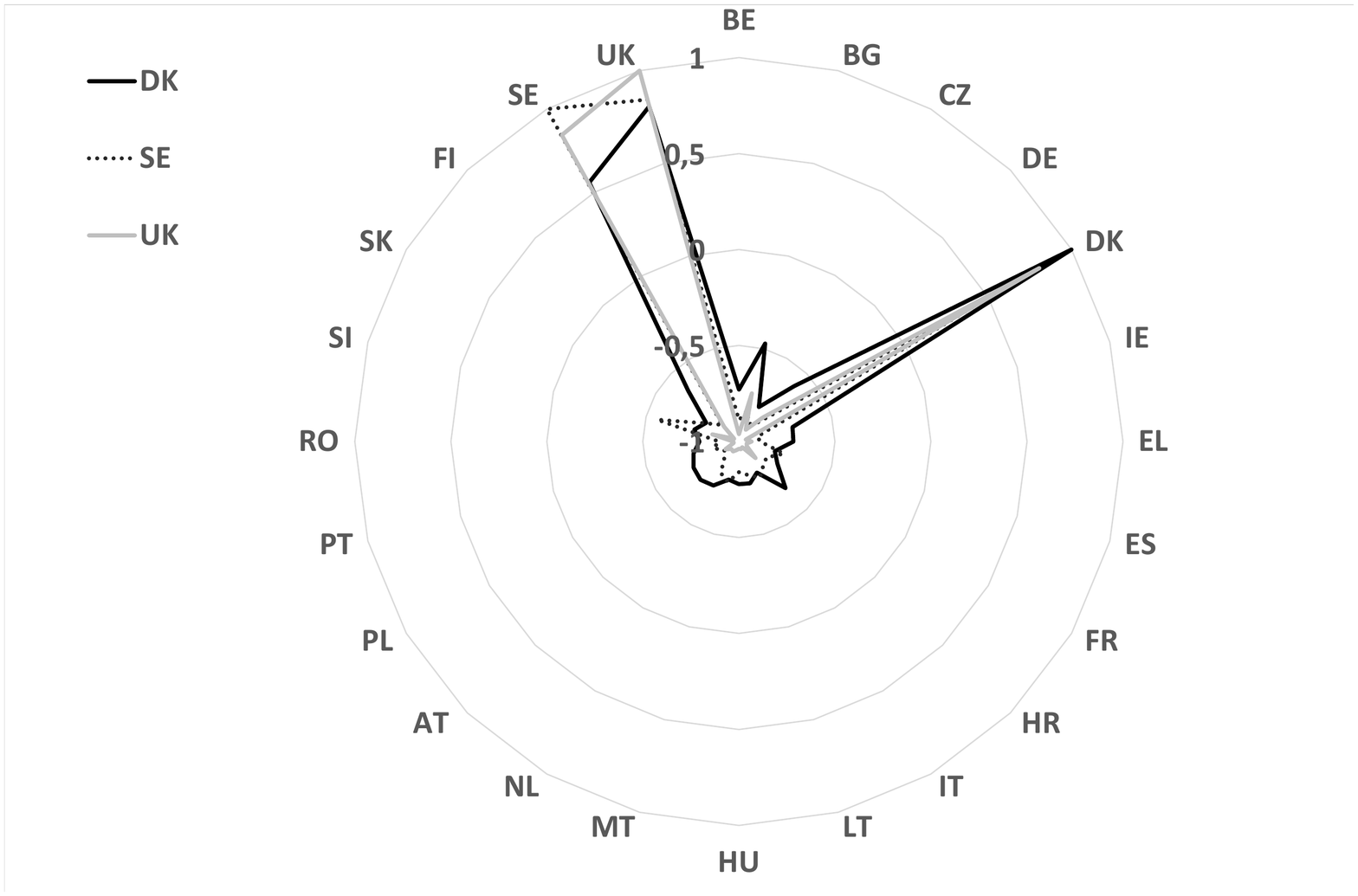}
\end{minipage}
\caption{Correlation coefficient computed on observed and simulated data: Denmamrk, Sweden and United Kingdom}\label{corr1}
\begin{minipage}{0.49\textwidth}
\includegraphics[width=1.1\textwidth]{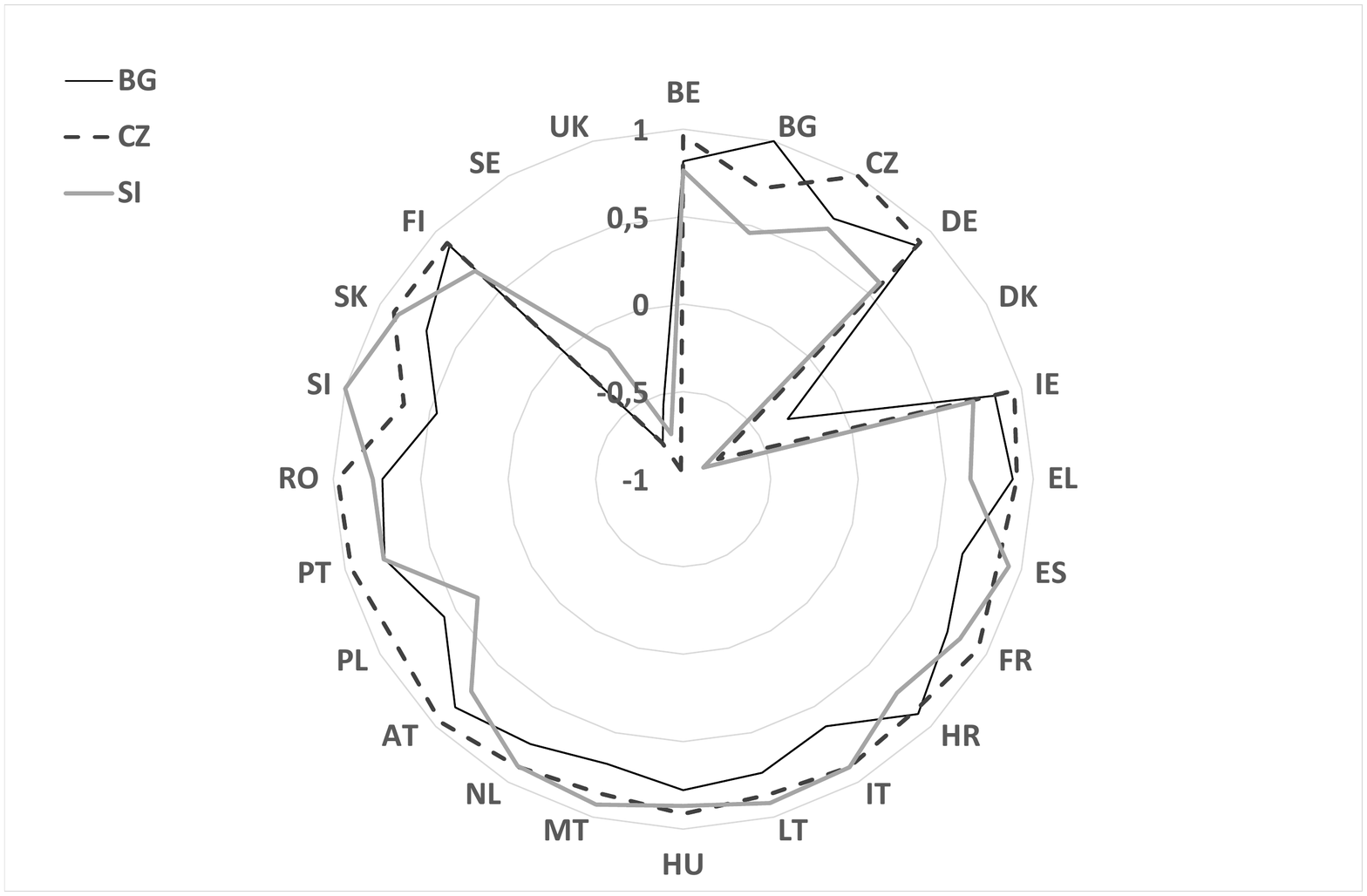}
\end{minipage}
\hspace{0.2cm}
\begin{minipage}{0.49\textwidth}
\includegraphics[width=1.1\textwidth]{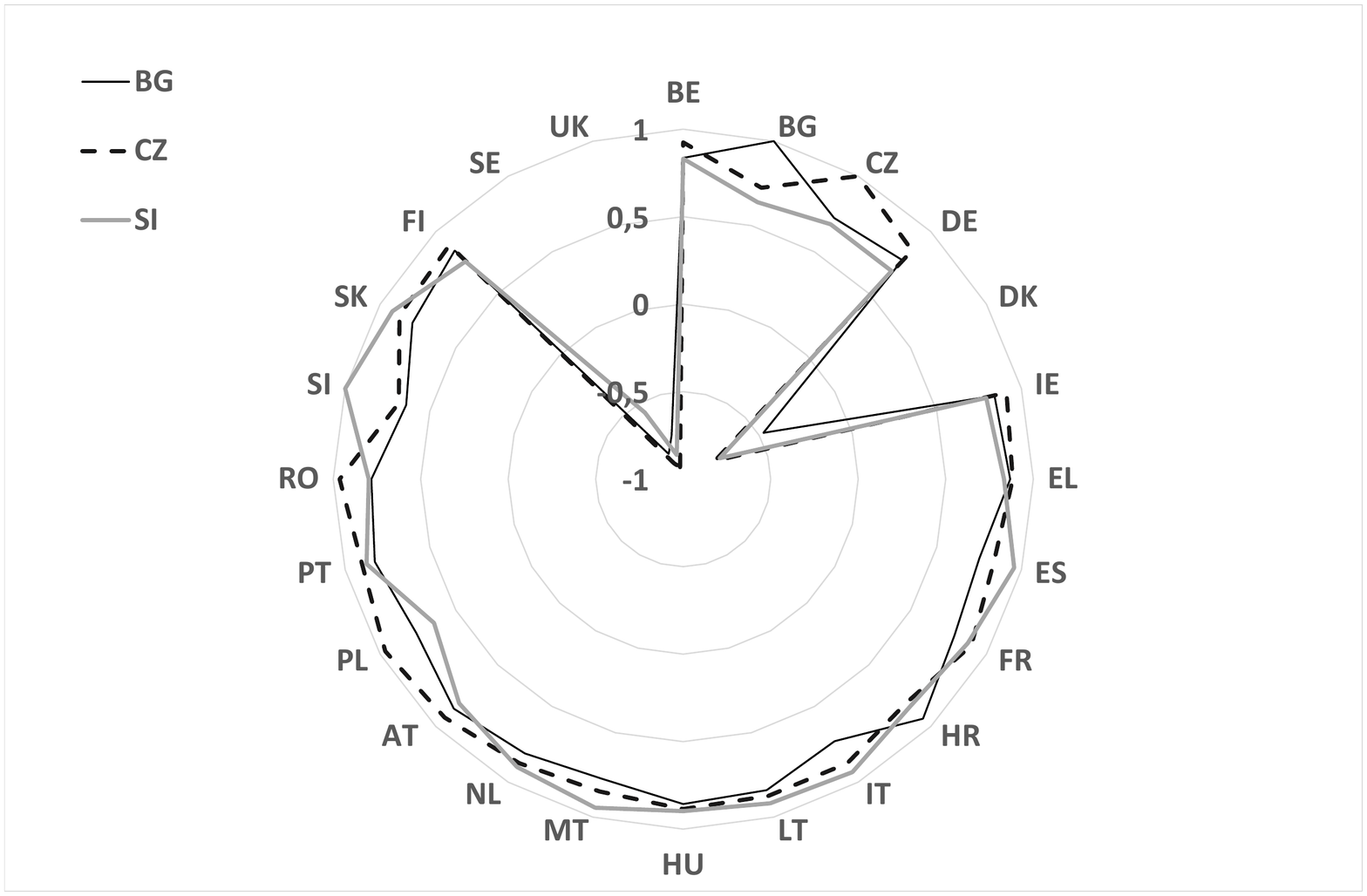}
\end{minipage}
\caption{Correlation coefficient computed on observed and simulated data: Bulgaria, Czech Republic and Slovenia}\label{corr2}
\end{figure}
Figure (\ref{corr1}) illustrates the coefficients between Denmark, Sweden and United Kingdom and all other countries. As highlighted above, this group is negatively correlated with the rest of EU, except with its components. The structure of  $\rho_{j_\alpha, j_\beta}^{\alpha,\beta}(t)$ is close to that of  $\rho^{\alpha,\beta}$, and the value of the correlation coefficient are almost similar. As a matter of example, $\rho_{j_\alpha, j_\beta}^{\alpha,\beta}(t)=0.81$ and $\rho^{\alpha,\beta}=0.83$ between Denmark and United Kingdom or between United Kingdom and Italy $\rho_{j_\alpha, j_\beta}^{\alpha,\beta}(t)=-0.96$ and $\rho^{\alpha,\beta}=-0.93$.  The same is true  for the countries showed in Figure (\ref{corr2}): Bulgaria, Czech Republic and Slovenia. For instance, the correlation between Belgium and Bulgaria computed from the model is $0.84$, while that computed on the observed data is $0.82$. 
However, although the values are very close, there exist some exceptions. The difference between  $\rho_{j_\alpha, j_\beta}^{\alpha,\beta}(t)$  and $\rho^{\alpha,\beta}$ between Bulgaria and Denmark is $33.83\%$, while between Italy and Sweden is $33.35\%$.

This results denote that the correlation structure is well reproduced by the model. Moreover, there is a very strong positive correlation between the most of countries  while a strong negative correlation would characterize the dependence structure between Denmark, Sweden and United Kingdom with respect to all other countries.

\section{Concluding Remarks}
In this paper, we propose a copula based Markov reward approach to investigate the financial risk in European Union. 
The novelty of this approach consists on the use of the piecewise homogeneous Markov chain to describe rating dynamics 
 and on  the inclusion of a stochastic process describing the spread evolution and its dependence among countries.
In particular the financial risk is evaluated by focusing  on its distribution by means of generalized measure of inequality and on its total amount. The latter is analysed by computing the total financial risk in a recursive way including every possible evolution of countries' rating assignment. The methodology proposed has been applied to real data concerning sovereign credit rating assigned by Moody's, Fitch and S\&P and sovereign credit spreads.  
Obtained results suggest that the financial risk will be less equi-distributed over the next future. The dynamic inequality shows an increasing trend, mostly driven by the intra-inequality measure, at a different rate for different rating agencies.  Differences are highlighted also on the evolution of the total risk, whose speed of growth is increasing for S\&P and Fitch, but decreasing for Moody,s. 
Furthermore, the investigation of the dependence structure reveals a strong correlation between countries both on the historical data and on those simulated by our model.\\ 
The decomposition of the dynamic inequality into inter and intra components permits to assess the explanatory power of the rating variables and how it evolves over time. 
Investigating this topic, could be helpful for policies aiming to control the financial risk within a given group of countries. Furthermore, the methodology is a potential tool to be applied to other financial problems such as portfolio loan management having the purpose to reduce the credit risk exposure.  
Finally, we think that possible extensions consist in the introduction of other economic and financial variables to give a better explanation  about the financial risk, in terms of its inequality and its quantification, by adding other sources of dependence.

\bibliography{mybib}

\bibliographystyle{plain}    
\end{document}